\documentclass[11pt,a4paper]{article}

\usepackage[T1]{fontenc}
\usepackage[utf8]{inputenc}
\usepackage{lmodern}
\usepackage[margin=25mm]{geometry}
\usepackage{microtype}
\usepackage{amsmath,amsfonts,amssymb,amsthm,mathtools,bm,bbm}

\usepackage{graphicx,eurosym}
\usepackage{booktabs,array,tabularx,multirow,arydshln}
\usepackage{float,algorithm,algpseudocode,adjustbox}
\usepackage{xcolor}
\usepackage[section]{placeins}
\usepackage[authoryear,round]{natbib}
\usepackage{xurl}
\usepackage[font=small,labelfont=bf,labelsep=period]{caption}
\usepackage{titlesec}
\usepackage[hidelinks,unicode,bookmarksnumbered]{hyperref}
\usepackage{bookmark}

\graphicspath{{figures/}}
\DeclareUnicodeCharacter{20AC}{\euro}
\hypersetup{
  pdftitle={Tax Fraud Detection via Bayesian Positive--Unlabeled Learning with Gaussian Processes and Multilayer Networks},
  pdfauthor={Konstantinos Bourazas; Angelos Alexopoulos; Petros Dellaportas; Konstantinos Kalogeropoulos},
  pdfsubject={Bayesian positive-unlabeled classification with Gaussian processes and multilayer networks},
  pdfkeywords={Bayesian classification, Gaussian processes, MCMC, multilayer network, positive-unlabeled learning, product of experts}
}

\theoremstyle{plain}
\newtheorem{theorem}{Theorem}[section]
\newtheorem{proposition}[theorem]{Proposition}
\newtheorem{lemma}[theorem]{Lemma}

\theoremstyle{definition}
\newtheorem{remark}[theorem]{Remark}

\newtheorem{assumption}[theorem]{Assumption}

\newtheorem{dfn}{Definition}

\titleformat{\section}{\large\bfseries}{\thesection}{0.7em}{}
\titleformat{\subsection}{\normalsize\bfseries}{\thesubsection}{0.7em}{}
\titleformat{\subsubsection}{\normalsize\itshape}{\thesubsubsection}{0.7em}{}
\titlespacing*{\section}{0pt}{2.5ex plus 1ex minus .2ex}{1.2ex plus .2ex}
\titlespacing*{\subsection}{0pt}{2ex plus .8ex minus .2ex}{.8ex plus .2ex}
\allowdisplaybreaks[1]
\begin{document}

\begin{center}
{\LARGE\bfseries Tax Fraud Detection via Bayesian Positive--Unlabeled Learning with Gaussian Processes and Multilayer Networks\par}
\vspace{1.1em}
{\normalsize Konstantinos Bourazas\textsuperscript{1}\qquad Angelos Alexopoulos\textsuperscript{2}\par}
\vspace{0.25em}
{\normalsize Petros Dellaportas\textsuperscript{3,4}\qquad Konstantinos Kalogeropoulos\textsuperscript{5}\par}
\vspace{0.9em}
{\footnotesize
\textsuperscript{1}Department of Statistics and Actuarial-Financial Mathematics,\\
University of the Aegean, Samos, Greece\par
\vspace{0.25em}
\textsuperscript{2}Department of Economics, Athens University of Economics and Business, Athens, Greece\par
\vspace{0.25em}
\textsuperscript{3}Department of Statistical Science, University College London, London, United Kingdom\par
\vspace{0.25em}
\textsuperscript{4}Department of Statistics, Athens University of Economics and Business, Athens, Greece\par
\vspace{0.25em}
\textsuperscript{5}Department of Statistics, London School of Economics, London, United Kingdom\par}
\end{center}

\begin{abstract}
Tax fraud remains a central challenge for public revenue authorities worldwide, imposing fiscal losses estimated to reach up to~\EUR{1}~trillion annually in the EU alone. Fraudulent firms intentionally manipulate reported figures to conceal their activity. Business networks can provide complementary information and reveal fraud that covariates alone may miss. We propose a Bayesian positive-unlabeled (PU) classification framework that combines firm-level covariates with multilayer network information. As a policy relevant case study, we apply the framework to \textcolor{black}{firms in a regulated segment of the Greek energy market}, a sector exposed to excise tax evasion. Confirmed fraud labels exist only for audited cases, while the remaining observations are unlabeled rather than verified compliant. We develop a Bayesian Gaussian process (GP) classifier integrating covariates with multilayer network information through a Product of Experts (PoE) construction and incorporating a nondetection probability for missed positives. \textcolor{black}{We derive identification bounds for the nondetection rate under an anchor condition and show that, for a fixed latent risk function, a common nondetection rate affects calibration but not ranking.} Simulations show strong ranking performance relative to PU and network-based alternatives, while illustrating the difficulty of estimating the nondetection rate. In real audit data, the method identifies high risk firms with posterior uncertainty, estimates the number of undetected fraudulent cases, and identifies whether risk is driven by covariates, networks, or both. \textcolor{black}{Of the six highest ranked firms, the two that had already been reinspected independently were both confirmed as noncompliant, providing an independent validation of these cases. The remaining four were selected by the tax authority for follow-up inspection.}
\end{abstract}


\section{Introduction} \label{sec:intro}
 
Tax fraud and evasion, often facilitated by complex schemes involving underreporting, shell companies, and illicit trade, undermine governments' capacity to fund essential public services such as healthcare, education, and infrastructure. In the European Union alone, annual tax losses have been estimated at up to~\EUR{1}~trillion \citep{european_commission_2017,european_commission_2022}. Excise duties on high value, easily traded goods such as alcohol, tobacco, and fuel are especially exposed, with evasion often relying on underreporting and coordinated schemes among connected firms \citep{ramos2024impact}. Tax authorities increasingly rely on statistical and machine learning methods to detect such activity, but their efforts are often constrained by severe class imbalance and limited reliable labels. Recent studies use machine learning to improve administrative audit selection \citep{battaglini2024refining}, semi supervised and unsupervised methods to address label scarcity \citep{kleanthous2020gated,vanhoeyveld2020value}, and taxpayer networks to detect VAT fraud \citep{alexopoulos2025network}. Audit resources are limited, so fraud is confirmed for only a small share of firms, while many cases remain undetected, which limits standard supervised learning in tax and customs enforcement. Networks offer a complementary source of evidence, since a firm controls the figures it reports but not the counterparties it transacts with. \textcolor{black}{Such relational information may arise from sources such as transaction links, ownership relations, shared business services, or geographic proximity.}
 
Greece provides a case study with direct policy relevance. \textcolor{black}{The Independent Authority for Public Revenue (IAPR) operates nationwide monitoring systems, recording electronically daily firm activity such as purchases and sales, together with administrative information on firm characteristics, audit histories, and interfirm relationships such as common board members, ownership links, and transaction networks.} Two main challenges arise. First, the few confirmed fraud cases give reliable positive labels, but firms without detected fraud cannot be assumed compliant, since they may not have been audited or fraudulent activity may be concealed or be missed during an audit, so they are treated as unlabeled rather than verified negatives. Second, reported variables \textcolor{black}{can be manipulated to evade the controls embedded in the electronic systems of the authorities}. Such behavior may appear plausible at the firm level but suspicious when viewed through the network of economic relationships.

Graph neural networks have become a prominent approach for combining node attributes with relational structure \citep{zhou2020graph}, learning representations that exploit both sources jointly. Their performance, however, typically depends on large volumes of reliable labels and extensive tuning, conditions rarely met in tax enforcement. Gaussian process (GP) classifiers instead provide a flexible nonparametric approach for learning complex decision boundaries with quantified uncertainty \citep{williams1998bayesian}. These formulations operate on covariates in Euclidean space, while GPs have since been extended to networks. Other examples are the Graph Gaussian Process of \citet{ng2018ggp}, which averages latent functions over one hop neighborhoods and learns from few labels, and Matérn type kernels on discrete spaces \citep{borovitskiy2021matern}, which enable relational smoothing on graphs.

Most GP approaches assume fully reliable labels or treat covariates and network data in isolation. Positive Unlabeled (PU) learning addresses the first limitation in settings where only a subset of positives is labeled and the remaining cases may belong to either class \citep{bekker2020learning}. Under the common Selected Completely At Random (SCAR) assumption, labeled positives form a random sample of all true positives \citep{elkan2008learning}, so treating unlabeled cases as negatives introduces systematic label noise. Much of the PU literature addresses this through risk correction. Unbiased PU methods recover the supervised classification risk \citep{duplessis2015convex}, while the nonnegative correction of \citet{kiryo2017nnpu} prevents the estimated negative class risk from becoming negative, thereby reducing severe overfitting with flexible classifiers. Other strategies include Bayesian formulations \citep{he2020bayesianPU,wang2023aoas} and AdaSampling \citep{yang2018adasampling}, which iteratively resamples a candidate negative training set from the unlabeled cases using the current classifier probabilities. Few methods combine an explicit model of the one sided labeling mechanism with heterogeneous covariate and network information in a fully Bayesian framework.

The structure of the application suggests that these two challenges should be addressed together. Covariates and network layers provide distinct but complementary evidence about the same latent fraud status: reported firm characteristics describe the unit itself, while information on relations may reveal suspicious patterns that are not visible from those characteristics alone. It is therefore natural to retain source-specific components rather than collapse all information into a single representation, while allowing evidence from the different sources to complement one another. At the same time, the small number of confirmed positives and the uncertainty of the remaining labels make an additional observation-specific gating mechanism difficult to learn reliably. This motivates a parsimonious construction in which all sources contribute to a common latent risk, while their relative influence is allowed to adapt through source-specific GP components. Therefore, we propose a Bayesian PU classification model for fraud detection that combines covariates with multilayer network data through a Product of Experts (PoE) construction. Each source contributes a separate GP expert, and the expert probabilities are combined multiplicatively. While PoE has traditionally been used to scale GPs through distributed or local experts \citep{liu2020gaussian}, we repurpose it to integrate heterogeneous sources of evidence about a common latent class under one sided label uncertainty.

Our contributions are threefold. First, we formulate a single model in which the covariate expert, the network experts, their mean functions, and the false negative rate are estimated jointly. It yields posterior fraud probabilities for unlabeled firms and a posterior distribution for the number of undetected fraud cases. The model also decomposes each firm's latent log odds exactly into a baseline term and source specific contributions, showing what drives its risk assessment. Second, we derive identification bounds for the false negative rate under an anchor condition and show that, for a fixed latent risk function, a common false negative rate rescales the scores by the same positive factor and therefore preserves the ROC curve and AUC. Third, the model admits an exact single GP representation with additive means and covariances. We exploit this to adapt the auxiliary gradient sampler of \citet{titsias2018auxiliary}, updating the latent field jointly with all kernel hyperparameters. Conditional on the aggregate covariance, the cost of the latent field refresh does not grow with the number of experts.

\textcolor{black}{In our application to audit data from a regulated segment of the Greek energy market, the model identifies high risk firms, including some for which relational information carries most of the evidence, with multiple relational layers providing complementary signals. Of the six highest ranked firms, two had already been reinspected independently by the tax authority and were confirmed as noncompliant, providing an independent retrospective validation of these cases. The remaining four were selected by the authority for follow-up inspection after it reviewed our results.} Simulations under both random and covariate dependent labeling show strong ranking performance relative to PU and network based alternatives, together with quantified posterior uncertainty.
 
The rest of the paper is organized as follows. Section~\ref{sec:model} introduces the PU--PoE model and the design-referenced prior calibration. Section~\ref{sec:eta} studies the identification of the noise rate
and posterior inference for the number of hidden positives, while
Section~\ref{sec:comp} presents the sampler.  Sections~\ref{sec:sim} and~\ref{sec:appl} present the simulation study and the tax-audit application, and Section~\ref{sec:disc} concludes this work. Complete proofs, prior derivations, computational details, and additional empirical results are provided in Appendices~\ref{supp:proofs}--\ref{supp:application}.

\section{The proposed modeling framework} \label{sec:model}

\subsection{Product of experts likelihood}
\label{subsec:poe-model}

Let $T_i\in\{0,1\}$ denote the true and generally unobserved class of unit $i$, for $i=1,\ldots,n$, with $T_i=1$ indicating a fraudulent unit. We observe $m+1$ sources of information, a vector of covariates and $m$ network layers. Each source is modeled by a Gaussian process expert $\mathcal M_j$, for $j=0,\ldots,m$, where $j=0$ is the covariate expert and $j\ge1$ are the network experts. Expert $\mathcal M_j$ carries a latent vector $\mathbf x_j=(x_{j,1},\ldots,x_{j,n})^\top$, and we write $p_{j,i}=\sigma(x_{j,i})$ with $\sigma(x)=\{1+\exp(-x)\}^{-1}$, so that $x_{j,i}$ is the evidence of source $j$ about unit $i$ on the log odds scale.

The Product of Experts construction \citep{hinton2002training} combines the source specific Bernoulli experts through normalized
product pooling with $\Pr(T_i=1\mid x_{0:m,i})\propto\prod_{j=0}^{m}p_{j,i}$ and $\Pr(T_i=0\mid x_{0:m,i})\propto\prod_{j=0}^{m}(1-p_{j,i})$. We add a term $\beta_0$ to account for the imbalance due to low rate of fraud, which is analogous to $\log\{\Pr(T_i=1)/\Pr(T_i=0)\}$, so that the PoE likelihood becomes
\begin{equation}
\label{eq:poe-prob}
\Pr(T_i=1\mid x_{0:m,i},\beta_0)
=
\frac{\exp(\beta_0)\prod_{j=0}^{m}p_{j,i}}
     {\exp(\beta_0)\prod_{j=0}^{m}p_{j,i}
      +\prod_{j=0}^{m}(1-p_{j,i})},
\qquad i=1,\ldots,n.
\end{equation}
The intercept $\beta_0$ can be read as one additional constant expert, so that $\sigma(\beta_0)$ is the fraud risk of a unit for which every source contributes zero. 

\begin{remark}
\label{lemma:poe-logistic}
Let $f_i=\beta_0+\sum_{j=0}^{m}x_{j,i}$ and $\mathbf f=(f_1,\ldots,f_n)^\top$. Then \eqref{eq:poe-prob} simplifies to
$\Pr(T_i=1\mid x_{0:m,i},\beta_0)= \Pr(T_i=1\mid f_i)=\sigma(f_i)$.
\end{remark}

Remark~\ref{lemma:poe-logistic} is the additive log odds form of product pooling. The heterogeneous sources combine into a single Bernoulli model whose latent logit is a common baseline plus one term
per source, and we write $p_i=\sigma(f_i)$ for the resulting fraud probability of unit $i$. \textcolor{black}{Equivalently, at the latent logit level the construction is an additive GP classifier, with one GP component for each information source.} Beyond this convenient representation, product pooling combines the evidence of all experts into a single joint probability, rather than smoothing across them by averaging as a mixture of experts does. The relative influence of each source is learned through the amplitude and mean scale of its GP prior, introduced in Section~\ref{subsec:prior}. This is appropriate in our setting, where the sources provide complementary evidence about a common latent fraud status, rather than representing competing models from which one source is to be selected for each unit. 

\subsection{One sided label noise}
\label{subsec:pu-model}

The class $T_i$ is not observed. What is recorded is an audit label $Y_i\in\{0,1\}$, equal to one when unit $i$ has been confirmed as fraudulent and zero otherwise. Confirmation is reliable, so $Y_i=1$ implies $T_i=1$, but the converse fails, since fraud may go undetected either because the unit was never audited or because the audit did not uncover it. A zero label therefore carries no guarantee of compliance, and units with $Y_i=0$ are unlabeled rather than verified negatives.

We model this one sided labeling mechanism under the Selected Completely At Random assumption \citep{elkan2008learning,bekker2020learning}, in which a true positive fails to be recorded with a false negative rate $\eta\in[0,1)$ that does not depend on the unit, so that $\Pr(Y_i=1\mid T_i=1)=1-\eta$ and $\Pr(Y_i=0\mid T_i=1)=\eta$, while a true negative is never recorded as positive, $\Pr(Y_i=1\mid T_i=0)=0$. Marginalizing over the latent class then gives $\Pr(Y_i=1\mid f_i,\eta)=(1-\eta)p_i$ and $\Pr(Y_i=0\mid f_i,\eta)=1-(1-\eta)p_i$, so that the observed label log likelihood is
\begin{equation}
\label{eq:pu-loglik}
\mathcal L(\mathbf f;\eta)
= \log p(\mathbf y\mid\mathbf f,\eta)
=
\sum_{i=1}^{n}
\Big[
y_i\log\big((1-\eta)p_i\big)
+
(1-y_i)\log\big(1-(1-\eta)p_i\big)
\Big],
\end{equation}
where $\mathbf y=(y_1,\ldots,y_n)^\top$ collects the realized audit labels. Setting $\eta=0$ recovers the standard Bernoulli logistic likelihood, in which every zero is treated as a confirmed negative. For $\eta>0$ the zeros retain some probability of being undetected positives, which attenuates the evidence they provide against the positive class in proportion to $\eta$. The rate $\eta$ is not fixed in advance but estimated jointly with the latent field, and Section~\ref{sec:eta} shows what the observed labels can and cannot reveal about it.

\subsection{Prior specification}
\label{subsec:prior}

Starting from the experts, we assign each expert an independent Gaussian process prior with a source specific mean and covariance kernel. The covariance represents similarity at the natural scale of each source, Euclidean for covariates and relational for networks, while the experts are coupled through the common likelihood. Regarding the mean function, GP classifiers commonly use a zero mean prior. In our setting, however, labels are sparse, so the posterior latent function is pulled toward the prior mean in weakly informed regions, making the usual zero-mean specification potentially restrictive. We therefore use a structured mean to capture broad source specific patterns, leaving the covariance kernels to model residual nonlinear variation in covariate space and residual dependence over the network structure \citep{rasmussen2006gaussian}. 

For the covariate expert ($j=0$), let $\mathbf s_i\in\mathbb R^p$ denote the standardized covariates of unit $i$, stacked in the design matrix $\mathbf Z_0=[\mathbf s_1,\ldots,\mathbf s_n]^\top\in\mathbb R^{n\times p}$, and set
\[
\mathbf x_0 \sim \mathcal N(\boldsymbol\mu_0,\mathbf C_0),
\qquad
\boldsymbol\mu_0=\mathbf Z_0\boldsymbol\gamma_0,
\qquad
[\mathbf C_0]_{ij}
=
\sigma_0^2
\exp\!\left(-\frac{\|\mathbf s_i-\mathbf s_j\|^2}{2l_0^2}\right).
\]
The mean captures the linear covariate effect through the coefficients $\boldsymbol\gamma_0\in\mathbb R^p$, and the squared exponential covariance captures smooth nonlinear deviations from it, over a range of covariate distances set by $l_0$.

For each network layer $j\ge1$, let $\mathbf L^{(j)}_{\mathrm{norm}}=\mathbf U^{(j)}\boldsymbol\Lambda^{(j)}\mathbf U^{(j)\top}$ be the eigendecomposition of the normalized Laplacian of the layer. Following \citet{borovitskiy2021matern}, we use a spectral heat kernel built from this Laplacian, with a community informed mean,
\[
\mathbf x_j \sim \mathcal N(\boldsymbol\mu_j,\mathbf C_j),
\qquad
\boldsymbol\mu_j=\mathbf Z_j\boldsymbol\gamma_j,
\qquad
\mathbf C_j
=
\sigma_j^2
\mathbf U^{(j)}
\exp\!\Big(-\tfrac{1}{2l_j^2}\boldsymbol\Lambda^{(j)}\Big)
\mathbf U^{(j)\top},
\]
where $\mathbf Z_j\in\mathbb R^{n\times q_j}$ contains the leading eigenvectors of the modularity matrix of layer $j$ \citep{newman2006modularity}, retained up to the largest gap among its positive eigenvalues, and $\boldsymbol\gamma_j\in\mathbb R^{q_j}$ are the corresponding coefficients. Nodes with no edges receive an independent contribution with variance $\sigma_j^2\exp\{-1/(2l_j^2)\}$. Here $l_j$ sets how far risk propagates over the graph, with small values concentrating the prior variance on the low frequency eigenvectors and producing smooth global variation, and large values driving $\mathbf C_j$ towards $\sigma_j^2\mathbf I$.

For the mean coefficients we use a random effects formulation, $\gamma_{j,k}\mid\tau_j\sim\mathcal N(0,\tau_j^2)$ for $k=1,\ldots,q_j$, with $q_0=p$ for the covariate expert, and a block scale $\tau_j\sim\mathrm{Half\text{-}}t_{\nu}(0,\lambda_j)$, which shrinks weak blocks towards zero while its heavy tail lets strong ones escape. For the kernel parameters we set $\log l_j\sim\mathcal N(\mu_{l,j},s_{l,j}^2)$ and $\log\sigma_j\sim\mathcal N(\mu_{\sigma,j},s_{\sigma,j}^2)$, keeping both positive over a wide range of values. Finally, $\beta_0\sim\mathcal N(m_{\beta_0},s_{\beta_0}^2)$ and $\eta\sim\mathrm{Beta}(a_\eta,b_\eta)$, the natural choice on the unit interval. 
\textcolor{black}{We use a weakly informative prior for $\eta$, taking $a_\eta,b_\eta\geq1$, which by Remark~\ref{rem:eta-concave} keeps the full conditional log posterior concave. Within this range we set $a_\eta=2$, the smallest integer above one, while $b_\eta$ controls the upper quantiles of the prior. When the likelihood provides limited information about $\eta$, e.g., under violations of the SCAR assumption, this prior discourages posterior concentration near the boundary at zero. Based on the operational information discussed in Section~\ref{sec:data}, we choose $b_\eta$ so that the prior $95$th percentile equals $0.5$, which gives our main prior $\eta\sim\mathrm{Beta}(2,6.39)$ with mode $0.156$. In this way, we follow a fully Bayesian approach that incorporates the information available on the audit process.}

\textcolor{black}{We keep the same fully Bayesian approach for the GP experts, whose hyperparameters are calibrated from the covariate and network designs alone, before any labels are used, while prior information is incorporated for the remaining parameters when available. In general, we keep the priors weakly informative, since vague choices would leave the flat directions of the likelihood unresolved while sharp ones would displace the evidence in the labels.} Appendix~\ref{supp:prior-elicitation} gives the full elicitation procedure together with the values used in each analysis.

\subsection{Illustrative example}
\label{sec:illustration}
 
We use a controlled example to show what each component contributes. The data contain $n=100$ units, $70$ true negatives and $30$ true positives, of which $20$ are recorded and $10$ are hidden among the zero labels, so $\eta=1/3$. Each unit carries ten covariates drawn from $\mathcal N(0,1)$ for the negatives and $\mathcal N(0.5,1)$ for the positives, and the network is a two block stochastic block model on the same partition, with edge probabilities $0.10$ within the negatives, $0.16$ within the positives and $0.05$ across. The true classes only group the units in the figure and are never supplied to any model.
 
Four models are fitted, a covariates only model, a network only model, a PoE model combining the two experts, and the PU--PoE model that additionally estimates $\eta$. Priors are calibrated as described in Appendix~\ref{supp:prior-elicitation}. The latent field has median effective scale $2$ in every model, given to the single expert in the first two and split equally in the last two, so that the four are directly comparable. Each lengthscale prior spans the $10\%$ to $50\%$ points of the attainable average correlation range of its own design, $\beta_0$ places the baseline risk between $0.01$ and $0.10$, and $\eta\sim\mathrm{Beta}(2,6.39)$.
 
Figure~\ref{fig:illustration} reports posterior means and $80\%$ HPD intervals of $\sigma(f_i)$, with units ordered within each group by the PU--PoE posterior mean and the same ordering across all panels. The PoE model combines the evidence of the two experts, and the PU extension raises the probabilities that are already moderately high while leaving the low ones almost unchanged, so the hidden positives stand out without the ordering being disturbed. The posterior mean of $\eta$ is $0.230$ with $80\%$ HPD interval $(0.062, 0.357)$, against a prior mean of $0.238$ and a true value of $1/3$, illustrating that the noise rate is less precisely identified than the risk ordering. Appendix~\ref{supp:prior-sensitivity} \textcolor{black}{examines prior sensitivity over $100$ datasets generated from the same design, and Appendix~\ref{supp:illustration-experts} decomposes the PU--PoE fit into its two experts}.

\begin{figure}[tbp]
  \centering
  \includegraphics[width=0.925\linewidth]{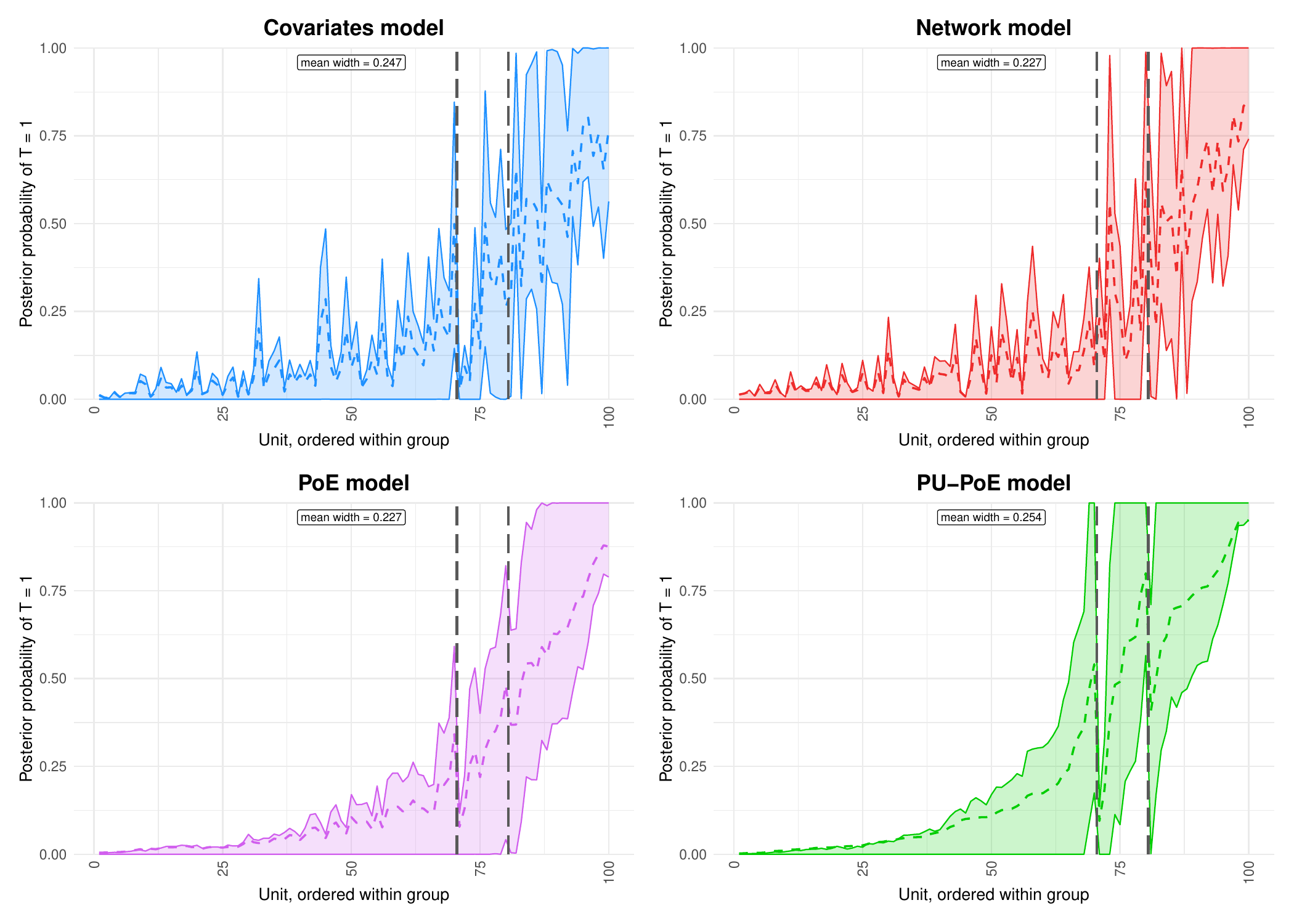}
\caption{Posterior means (dashed) and $80\%$ HPD intervals (shaded) of the latent positive probabilities under the four models. In each panel, the two vertical lines separate the true negatives (left), the hidden positives (middle) and the observed positives (right), and units are ordered within each group by the PU--PoE posterior mean, with the same order in all panels.}
  \label{fig:illustration}
\end{figure}

\section{Identification and inference under one-sided label noise}
\label{sec:eta}
 
Inference on the false-negative rate $\eta$ is important for calibrating latent fraud probabilities and quantifying undetected positives, as it governs the relationship between observed labels and latent fraud status. Let $q(x)=(1-\eta)\sigma(f(x))$ denote the observed conditional success probability. Identifiability of $(f,\eta)$ means that if two pairs $(f,\eta)$ and $(f_0,\eta_0)$ induce the same $q(x)$ almost everywhere, then they coincide $P_X$-almost everywhere. The result below bounds $\eta$ within an interval when the anchor holds approximately, and identifies it exactly when $\operatorname*{ess\,sup}_{x}\sigma(f(x))=1$. All essential suprema are taken with respect to $P_X$.

\begin{assumption}
\label{ass:lower-anchor}
There exists a known $\delta\in[0,1)$ such that
$\operatorname*{ess\,sup}_{x}\sigma(f(x))\ge1-\delta$.
\end{assumption}
 
\begin{lemma}
\label{lem:ID-ROC}
Suppose $Y\mid X\sim\mathrm{Ber}((1-\eta)\sigma(f(X)))$ with $\eta\in[0,1)$.
\begin{enumerate}
\item[(a)] Under Assumption~\ref{ass:lower-anchor},
\begin{equation}
\label{eq:eta-bounds}
\eta\in
\left[
\max\left\{0,1-
\frac{\operatorname*{ess\,sup}_xq(x)}{1-\delta}
\right\},
\;
1-\operatorname*{ess\,sup}_xq(x)
\right].
\end{equation}
When $\delta=0$, so that $\operatorname*{ess\,sup}_{x}\sigma(f(x))=1$, the two bounds coincide, and $\sigma(f(x))=q(x)/(1-\eta)$ $P_X$-almost everywhere.
\item[(b)] For fixed $f$, the score vectors
$\{\sigma(f_i)\}_{i=1}^n$ and
$\{(1-\eta)\sigma(f_i)\}_{i=1}^n$ induce the same ROC curve and
hence the same AUC.
\end{enumerate}
\end{lemma}
 
Part (a) brackets $\eta$ within an interval of width at most
$\operatorname*{ess\,sup}_{x}q(x)\,\delta/(1-\delta)$, which reduces
to a point when $\delta=0$. In the application, we examine whether
the fitted risks are compatible with the anchor condition. Part (b) shows that, for a fixed $f$, a single $\eta$ shared by all units changes calibration but not ranking.
 
Posterior draws of $(\mathbf f,\eta)$ also quantify the number of positives hidden among the zero-labeled units. Let $H=\sum_{i:y_i=0}\mathbf 1\{T_i=1\}$. For each zero-labeled unit, write
\[
h_i
=
\Pr(T_i=1\mid Y_i=0,f_i,\eta)
=
\frac{\eta\sigma(f_i)}{1-(1-\eta)\sigma(f_i)}.
\]
Hence the coherent posterior mean is
\begin{equation}
\label{eq:H-posterior}
\mathbb E(H\mid\mathbf y)
=
\sum_{i:y_i=0}
\mathbb E(h_i\mid\mathbf y).
\end{equation}
Posterior draws of $H$ that take integer values are obtained by drawing $T_i^{(s)}\sim\mathrm{Bernoulli}(h_i^{(s)})$ independently across units, conditional on each posterior draw, and summing over the units for which the observed label is zero.

\textcolor{black}{We also report the estimated number of hidden positives based on $\eta$ under SCAR. If $L=\sum_i y_i$ and $N^+=L+H$, then $L\mid \left( N^+,\eta \right) \sim\mathrm{Binomial}(N^+,1-\eta)$. Replacing $N^+$ by $L/(1-\eta)$ gives, for posterior draw $s$, $H_\eta^{(s)}=L\frac{\eta^{(s)}}{1-\eta^{(s)}}$, with posterior mean
\begin{equation} \label{eq:H-estimator}
\widehat H_{\eta}=\frac{1}{S}\sum_{s=1}^S H_\eta^{(s)}.
\end{equation}
We distinguish this SCAR-implied count from the posterior quantity in \eqref{eq:H-posterior}, which uses the risk of each individual unit.} The latter is the coherent posterior target, whereas $\widehat H_\eta$ provides a simple diagnostic of the overall prevalence implied by the SCAR assumption.

\section{Posterior computation}
\label{sec:comp}

Let $\boldsymbol\theta=\{(\sigma_j,l_j)_{j=0}^m\}$ collect the kernel hyperparameters. By independence of the GP experts, the conditional prior of $\mathbf{f}$ is
\[
\mathbf f\mid \left( \beta_0,\boldsymbol\gamma,\boldsymbol\theta \right)
\sim
\mathcal N(\boldsymbol\mu_{\mathbf f},\mathbf C_{\mathbf f}),
\]
with $\boldsymbol\mu_{\mathbf f}
=
\beta_0\mathbf 1_n+\sum_{j=0}^m\mathbf Z_j\boldsymbol\gamma_j$ and
$\mathbf C_{\mathbf f}=\sum_{j=0}^m\mathbf C_j(\boldsymbol\theta_j)$. Thus, although the model is formulated in terms of multiple source-specific experts, posterior computation for the latent field can be carried out using a single GP with additive mean and covariance.
The joint posterior is
\begin{equation}
\label{eq:post-target}
\pi(\mathbf f,\boldsymbol\theta,\beta_0,\boldsymbol\gamma,
\boldsymbol\tau,\eta\mid\mathbf y)
\propto
p(\mathbf y\mid\mathbf f,\eta)
\mathcal N(\mathbf f\mid\boldsymbol\mu_{\mathbf f},\mathbf C_{\mathbf f})
\pi(\boldsymbol\theta)\pi(\beta_0)
\pi(\boldsymbol\gamma\mid\boldsymbol\tau)
\pi(\boldsymbol\tau)\pi(\eta).
\end{equation}
 
The sampler cycles through four blocks. The pair $(\mathbf f,\boldsymbol\theta)$ is updated jointly by an auxiliary-gradient move followed by an $\mathbf f$-only refresh at the accepted kernel, $(\beta_0,\boldsymbol\gamma)$ are drawn exactly from their Gaussian full conditional, each Half-$t$ scale $\tau_j$ and its inverse-gamma auxiliary are updated in closed form, and $\eta$ is updated by slice sampling on the logit scale. Appendix~\ref{supp:mcmc-details} gives the conditional formulae. Remark~\ref{rem:eta-concave} shows that the conditional posterior of $\eta$ has a unique mode.

\begin{remark}
\label{rem:eta-concave}
Under the PU likelihood in \eqref{eq:pu-loglik} and a
$\mathrm{Beta}(a_\eta,b_\eta)$ prior with
$a_\eta,b_\eta\ge1$, the full conditional log-posterior of
$\eta$ is strictly concave on $(0,1)$ for every finite
$\mathbf f$.
\end{remark}

The main computational challenge is updating the high dimensional latent field $\mathbf f$ under the non-Gaussian PU likelihood. For $p_i=\sigma(f_i)$, the log likelihood gradient has entries
\begin{equation}
\label{eq:pu-gradient}
[\nabla_{\mathbf f}\mathcal L(\mathbf f;\eta)]_i
=
y_i(1-p_i)
-(1-y_i)
\frac{(1-\eta)p_i(1-p_i)}
{1-(1-\eta)p_i},
\end{equation}
which reduces to $y_i-p_i$ when $\eta=0$.

\begin{remark}
\label{rem:geometry-f}
If $\eta>0$ and at least one $y_i=0$, the PU log-likelihood is not globally concave in $\mathbf f$.
\end{remark}

Since the target is not necessarily log-concave, we use the auxiliary-gradient sampler aGrad-$z$ of \citet{titsias2018auxiliary}, which exploits the Gaussian prior of $\mathbf f$ while following the local gradient of the likelihood. Writing $\mathbf h=\mathbf f-\boldsymbol\mu_{\mathbf f}$ for the centered field and fixing a stepsize $\delta>0$, one update draws an auxiliary vector and then a proposal,
\begin{equation}
\label{eq:agrad}
\mathbf z\sim
\mathcal N\!\Big(
\mathbf h+\dfrac{\delta}{2}\nabla\mathcal L(\mathbf f;\eta),
\ \dfrac{\delta}{2}\mathbf I
\Big),
\qquad
\mathbf h'\sim
\mathcal N\!\Big(
\dfrac{2}{\delta}\mathbf A\mathbf z,
\ \mathbf A
\Big),
\qquad
\mathbf A=\Big(
\mathbf C_{\mathbf f}^{-1}
+\dfrac{2}{\delta}\mathbf I
\Big)^{-1}.
\end{equation}
The Gaussian prior cancels, and the proposal is accepted with probability
\begin{equation}
\label{eq:agrad-acceptance}
\alpha_{\mathrm{MH}}^{(\mathbf f)}
=
\min\left\{
1,
\exp\left[
\mathcal L(\mathbf f';\eta)-\mathcal L(\mathbf f;\eta)
+
g(\mathbf z,\mathbf h';\mathbf f')
-
g(\mathbf z,\mathbf h;\mathbf f)
\right]
\right\},
\end{equation}
where
$g(\mathbf z,\mathbf u;\mathbf v)
=
\{\mathbf z-\mathbf u-(\dfrac{\delta}{4})
\nabla\mathcal L(\mathbf v;\eta)\}^{\top}
\nabla\mathcal L(\mathbf v;\eta)$.

\begin{proposition}
\label{prop:poe-cost}
Conditionally on $(\beta_0,\boldsymbol\gamma,\boldsymbol\theta,\eta)$, the aGrad-$z$ update of $\mathbf f$ is valid for the PU--PoE model. Once $\mathbf C_{\mathbf f}$ has been formed and factorized, the cost of updating $\mathbf f$ is independent of the number of experts.
\end{proposition}

For fixed kernel hyperparameters one update costs $\mathcal O(n^2)$, after an $\mathcal O(n^3)$ eigendecomposition of $\mathbf C_{\mathbf f}$ that is computed once and reused. This applies only to $\mathbf f$, while learning $\boldsymbol\theta$ does depend on $m$. Every proposal rebuilds the kernels it moves and calls for a fresh eigendecomposition of $\mathbf C_{\mathbf f}$.

Since $\mathbf f$ and $\boldsymbol\theta$ are strongly dependent a posteriori, we update them jointly, proposing a new $\boldsymbol\theta$ and a new $\mathbf f$ together and accepting or rejecting the pair as one. Warm-up comes first and has two stages. The initial stage moves one expert at a time, which is slow but stable, and its purpose is to find the region the chain will occupy and to learn the scale of each block. The second stage uses the path just traced to build a joint proposal for all kernel parameters at once, and tunes only its overall size. Then, sampling proceeds with the proposals established during warm-up. We target acceptance rates near $55\%$ for $\mathbf f$ and $25\%$ for $\boldsymbol\theta$. Appendix~\ref{supp:mcmc-details} gives the proposal geometry and the acceptance ratios, with the joint move in Algorithm~\ref{alg:joint} and the schedule in Algorithm~\ref{alg:schedule}.

Individual experts follow by Gaussian conditioning on $\mathbf f$,
which requires only $\mathrm{Cov}(\mathbf x_j,\mathbf f)=\mathbf C_j$, by prior independence:
\[
\mathbf x_j
\mid
\left(\mathbf f,\boldsymbol\theta,\boldsymbol\gamma,\beta_0\right)
\sim
\mathcal N\!\left(
\boldsymbol\mu_j+
\mathbf C_j\mathbf C_{\mathbf f}^{-1}
(\mathbf f-\boldsymbol\mu_{\mathbf f}),
\;
\mathbf C_j-
\mathbf C_j\mathbf C_{\mathbf f}^{-1}\mathbf C_j
\right).
\]
Together with the intercept, the conditional means add to $\mathbf f$, so the decomposition is exact on the logit scale. We report each expert as $\sigma(x_{j,i})$, which exceeds one half when
the contribution is positive.

\section{Simulation Study}
\label{sec:sim}

We evaluate the proposed model in three designs, built by combining two data generating mechanisms with two mechanisms for hiding positives. In the first, called linear, the covariates are linearly separable and the network is drawn from a two block Stochastic Block Model (SBM). In the second, called nonlinear, the covariates carry interactions and curved boundaries that no hyperplane separates, and the network comes from a composition of three Chinese Restaurant Processes (CRP). Regarding hidden positives, they are hidden either completely at random (SCAR), which is the mechanism the model assumes, or with a probability that depends on the covariates (SAR), under which the least suspicious positives go unrecorded first and the assumption is violated in a direction that is particularly challenging to the model. \textcolor{black}{The main text reports the nonlinear design under both SCAR and SAR, while Appendix~\ref{supp:sim} reports the linear design under SCAR.} Each design is replicated $100$ times. A replication has $n=400$ units, $p=10$ covariates and one network layer, with $45$ true positives and $355$ true zeros. Of the positives, $30$ are recorded and $15$ are hidden among the zero labels, so the true nondetection rate is $\eta=1/3$. To stress the model, the positives are heterogeneous by construction, one third carrying covariate signal only, one third network signal only and one third both, with hiding applied within each group. Thus, the experts are in mild disagreement on the positives by construction, as may also occur in practice, so not all hidden positives receive support from both sources and complete recovery is not expected.

For the nonlinear design, covariates are generated class
conditionally in five bivariate blocks that distinguish the $30$ units that carry a covariate signal from the remaining $370$, each block carrying a different kind of signal. Specifically, the five blocks comprise a mean shift, two separated clusters placed symmetrically on opposite sides of the background distribution so that their pooled mean matches the background mean, a noisy ring within a disc, a displacement from a common parabolic trend, and identical marginals with opposite within pair correlations. Only the first provides a linear mean shift, while the remaining blocks encode nonlinear or dependence based signal. The network is constructed by combining three edge layers based on Chinese Restaurant Process (CRP) partitions, one over all units, one over the units carrying network signal and one over the remainder. All three use the same concentration and expected within community degree, and their union is followed by one pass of triadic closure. Accordingly, the class signal lies in the pattern of connections rather than in node degree. The exact distributions and network parameters are given in Appendix~\ref{supp:nonlinear}.

The nested Bayesian comparison isolates one component at a time.
Models with the structured mean alone are denoted Lin and those that
add the Gaussian process residual GP, while Cov, Net and PU indicate
which of the covariates, the network and the noise mechanism are
included, with PoE denoting the two experts combined. We also fit
PU-PoE under $\boldsymbol\mu=\beta_0\mathbf 1_n$, which drops the
structured mean and serves as its baseline. The remaining competitors
cover a range of approaches to label noise, class
imbalance and network structure. AdaSampling (Ada;
\citealp{yang2018adasampling}) with an SVM base learner represents PU
learning through adaptive resampling and margin based classifiers,
and the non-negative positive unlabeled risk estimator (nnPU;
\citealp{kiryo2017nnpu}) the risk correction approach. We also fit
weighted logistic regression (WLR), with a weight of $1/2$ on the
zero labeled units, and Bayesian logistic regression fitted through
the Integrated Nested Laplace Approximation (INLA;
\citealp{rue2009approximate}). A two layer Graph Convolutional
Network (GCN; \citealp{KipfWelling2017GCN}) serves as a deep learning benchmark that incorporates relational structure directly. The first four are also fitted in network augmented versions using the same modularity eigenvectors, selected at the largest eigengap among the positive eigenvalues as in Section~\ref{subsec:prior}. \textcolor{black}{Implementation settings for the competing methods and the priors for the linear Bayesian models are given in Appendix~\ref{supp:competitors}.}

For the GP models, the relevant priors are recalibrated for each replication through Algorithm~\ref{alg:hyperprior}. The total marginal standard deviation of the residual field is centered at $c=2$ and the variance is split equally between the two experts, so each carries median $\sqrt2$ with a $95\%$ interval of $[0.71,2.83]$, obtained by halving and doubling the center. Each lengthscale is calibrated so that the induced correlation covers $10\%$ to $50\%$ of its attainable range with probability $0.95$. The structured mean uses Half-$t_3$ scales set by the median matching rule of Section~\ref{subsec:prior}. The intercept is $\beta_0\sim\mathcal N(-3.40,0.61^2)$, which places a $95\%$ prior interval of $[0.01,0.10]$ on the prevalence, and $\eta\sim\mathrm{Beta}(2,6.39)$.

For the Bayesian methods, we first run five independent pilot chains of $40{,}000$ iterations on separate samples, to obtain the starting values, the step size and the proposal geometry used in every replication. The production chain then runs
$15{,}000$ warm-up iterations, the first third blocked by expert and the rest joint, followed by $40{,}000$ sampling iterations thinned by $20$, giving $2{,}000$ stored draws with every proposal parameter frozen at its final warm-up value.

Overall AUC and Overall Recall@45 use all $400$ units and the latent classes $T_i$, while their hidden counterparts are restricted to the $370$ unlabeled units and to the $15$ positives among them. Writing $S_k$ for the $k$ highest ranked units, $\text{Recall@}45=|S_{45}\cap\{T_i=1\}|/45$ over all units and $\text{Recall@}15=|S_{15}\cap\{y_i=0,T_i=1\}|/15$ over the unlabeled ones. Log loss and Brier scores are computed against $T_i$ from the same probabilities that produce the rankings, while coverage and interval scores \citep{gneiting2007strictly} are at the stated level and are available only for methods that produce predictive draws for the latent labels. \textcolor{black}{For the Bayesian models, these are generated as $T_i^{(s)}\sim\mathrm{Bernoulli}\{\sigma(f_i^{(s)})\}$.} Also, for the two best PU models we report the mean estimate, the root mean squared error, the coverage and the interval score of $\eta$, and the same four for the implied number of undetected positives $H$. 

Table~\ref{tab:sim-scar} presents the results for the nonlinear design under SCAR. Using a Gaussian process for the covariates instead of a linear mean raises the overall AUC from $0.659$ to $0.907$, the hidden AUC from $0.566$ to $0.765$ and Recall@45 from $0.292$ to $0.705$. Adding the network expert raises these to $0.949$, $0.856$ and $0.780$, so the two sources contribute in that order and neither is redundant. \textcolor{black}{Consistent with Lemma~\ref{lem:ID-ROC}(b), the PU correction has little effect on ranking but substantially affects the predicted probabilities.} Specifically, across the four pairs of models that differ only in the correction the overall AUC moves by at most $0.001$, while the hidden log loss falls by between $0.193$ and $0.278$ and the hidden Brier score by between $0.030$ and $0.071$. The same split appears between the structured mean and its
intercept only baseline. The baseline ranks marginally better, at $0.951$, $0.862$ and $0.784$ against $0.948$, $0.854$ and $0.780$, while the structured mean is better calibrated, at $0.184$ against $0.187$ in overall log loss and $0.052$ against $0.053$ in overall Brier. The two separate more clearly on the noise rate itself. Table~\ref{tab:eta-nonlinear-scar} reports the mean estimate, the root mean squared error, the coverage and the interval score of $\eta$, together with the same quantities for the implied number of undetected positives $H$. The full model gives a mean $\hat\eta$ of $0.187$ with coverage $0.470$, while GP-PU-PoE$_{\beta_0}$ gives $0.177$ with coverage $0.380$, and the same ordering holds for $H$ at $6.955$ against $6.500$. Both models substantially underestimate the true value $\eta=1/3$, and the corresponding intervals have below-nominal coverage. This reflects the weak identification of the noise rate when latent risks do not approach one sufficiently closely, as anticipated by Lemma~\ref{lem:ID-ROC}(a); inference on $\eta$ and the number of hidden positives is therefore appreciably more difficult than risk ranking in this design.

\begin{table}[!t]
\centering
\caption{Nonlinear design under SCAR. Monte Carlo means over $100$
replications, standard deviations in parentheses. Best value in bold,
closest to nominal for coverage.}
\label{tab:sim-scar}
\begingroup
\scriptsize
\setlength{\tabcolsep}{1.4pt}
\renewcommand{\arraystretch}{1.1}
\setlength{\aboverulesep}{0.15ex}
\setlength{\belowrulesep}{0.15ex}

\resizebox{\linewidth}{!}{
\begin{tabular}{lcccccc}
\toprule
Method
& \shortstack{Overall\\AUC}
& \shortstack{Hidden\\AUC}
& \shortstack{Overall\\Recall@45}
& \shortstack{Hidden\\Recall@15}
& \shortstack{Overall\\LogLoss}
& \shortstack{Hidden\\LogLoss} \\
\midrule
Lin-Cov & 0.659 (0.043) & 0.566 (0.085) & 0.292 (0.058) & 0.110 (0.076) & 0.346 (0.010) & 2.668 (0.142) \\
Lin-PU-Cov & 0.659 (0.042) & 0.566 (0.084) & 0.292 (0.054) & 0.107 (0.076) & 0.338 (0.010) & 2.475 (0.146) \\
Lin-Cov+Net & 0.685 (0.057) & 0.596 (0.094) & 0.306 (0.083) & 0.111 (0.090) & 0.338 (0.017) & 2.631 (0.164) \\
Lin-PU-Cov+Net & 0.685 (0.057) & 0.597 (0.093) & 0.303 (0.083) & 0.106 (0.086) & 0.330 (0.018) & 2.418 (0.186) \\
\hdashline
GP-Cov & 0.907 (0.025) & 0.765 (0.059) & 0.705 (0.049) & 0.305 (0.093) & 0.228 (0.017) & 2.485 (0.188) \\
GP-PU-Cov & 0.907 (0.024) & 0.763 (0.059) & 0.706 (0.046) & 0.299 (0.099) & 0.221 (0.017) & 2.248 (0.200) \\
GP-PoE & 0.949 (0.018) & 0.856 (0.052) & 0.780 (0.042) & 0.392 (0.103) & 0.191 (0.014) & 2.278 (0.214) \\
GP-PU-PoE & 0.948 (0.018) & 0.854 (0.053) & 0.780 (0.042) & 0.393 (0.104) & \textbf{0.184} (0.015) & 2.000 (0.252) \\
GP-PU-PoE$_{\beta_0}$ & \textbf{0.951} (0.018) & \textbf{0.862} (0.052) & \textbf{0.784} (0.038) & \textbf{0.402} (0.103) & 0.187 (0.014) & 1.995 (0.226) \\
\hdashline
Ada & 0.824 (0.078) & 0.730 (0.079) & 0.548 (0.079) & 0.254 (0.103) & 0.448 (0.040) & 0.849 (0.124) \\
Ada+Net & 0.814 (0.109) & 0.713 (0.102) & 0.534 (0.098) & 0.214 (0.111) & 0.466 (0.044) & \textbf{0.847} (0.133) \\
nnPU & 0.565 (0.055) & 0.511 (0.085) & 0.342 (0.060) & 0.180 (0.078) & 0.541 (0.048) & 5.425 (0.666) \\
nnPU+Net & 0.593 (0.079) & 0.532 (0.093) & 0.377 (0.077) & 0.177 (0.083) & 0.515 (0.077) & 5.305 (0.707) \\
WLR & 0.648 (0.045) & 0.545 (0.089) & 0.291 (0.060) & 0.114 (0.081) & 0.337 (0.011) & 2.177 (0.256) \\
WLR+Net & 0.681 (0.056) & 0.579 (0.092) & 0.315 (0.071) & 0.122 (0.091) & 0.328 (0.018) & 2.163 (0.274) \\
INLA & 0.641 (0.047) & 0.537 (0.092) & 0.299 (0.060) & 0.115 (0.082) & 0.353 (0.012) & 2.984 (0.310) \\
INLA+Net & 0.674 (0.057) & 0.573 (0.094) & 0.322 (0.073) & 0.122 (0.093) & 0.346 (0.017) & 3.000 (0.342) \\
GCN & 0.821 (0.039) & 0.687 (0.090) & 0.486 (0.067) & 0.191 (0.096) & 0.286 (0.023) & 2.815 (0.462) \\
\bottomrule
\end{tabular}
}

\vspace{0.3ex}

\resizebox{\linewidth}{!}{
\begin{tabular}{lcccccc}
\toprule
Method
& \shortstack{Overall\\Cov$_{80}$}
& \shortstack{Hidden\\Cov$_{80}$}
& \shortstack{Overall\\IS$_{80}$}
& \shortstack{Hidden\\IS$_{80}$}
& \shortstack{Overall\\Brier}
& \shortstack{Hidden\\Brier} \\
\midrule
Lin-Cov & \textbf{0.916} (0.017) & 0.165 (0.132) & 0.938 (0.108) & 8.512 (1.186) & 0.098 (0.003) & 0.853 (0.018) \\
Lin-PU-Cov & 0.935 (0.016) & 0.298 (0.150) & 0.873 (0.097) & 7.318 (1.352) & 0.096 (0.003) & 0.823 (0.022) \\
Lin-Cov+Net & 0.923 (0.017) & 0.208 (0.143) & 0.896 (0.125) & 8.128 (1.283) & 0.096 (0.005) & 0.840 (0.035) \\
Lin-PU-Cov+Net & 0.941 (0.016) & 0.357 (0.161) & 0.833 (0.114) & 6.783 (1.450) & 0.094 (0.005) & 0.804 (0.048) \\
\hdashline
GP-Cov & 0.975 (0.005) & 0.465 (0.113) & 0.415 (0.053) & 5.812 (1.013) & 0.066 (0.005) & 0.795 (0.029) \\
GP-PU-Cov & 0.981 (0.005) & 0.567 (0.117) & 0.425 (0.055) & 4.894 (1.053) & 0.063 (0.005) & 0.746 (0.039) \\
GP-PoE & 0.981 (0.005) & 0.521 (0.127) & 0.339 (0.047) & 5.308 (1.145) & 0.056 (0.005) & 0.750 (0.043) \\
GP-PU-PoE & 0.986 (0.005) & 0.640 (0.129) & 0.345 (0.046) & 4.237 (1.165) & \textbf{0.052} (0.005) & 0.679 (0.064) \\
GP-PU-PoE$_{\beta_0}$ & 0.986 (0.005) & 0.646 (0.130) & \textbf{0.336} (0.048) & 4.186 (1.168) & 0.053 (0.005) & 0.685 (0.056) \\
\hdashline
Ada & -- & -- & -- & -- & 0.138 (0.016) & \textbf{0.314} (0.049) \\
Ada+Net & -- & -- & -- & -- & 0.146 (0.017) & 0.315 (0.056) \\
nnPU & -- & -- & -- & -- & 0.087 (0.009) & 0.911 (0.046) \\
nnPU+Net & -- & -- & -- & -- & 0.082 (0.009) & 0.907 (0.045) \\
WLR & 0.976 (0.007) & 0.683 (0.130) & 0.895 (0.082) & 3.856 (1.173) & 0.095 (0.004) & 0.732 (0.047) \\
WLR+Net & 0.977 (0.007) & \textbf{0.691} (0.136) & 0.846 (0.102) & \textbf{3.778} (1.226) & 0.093 (0.006) & 0.717 (0.056) \\
INLA & 0.937 (0.008) & 0.300 (0.122) & 0.852 (0.073) & 7.300 (1.097) & 0.095 (0.004) & 0.854 (0.035) \\
INLA+Net & 0.940 (0.008) & 0.335 (0.130) & 0.816 (0.083) & 6.988 (1.172) & 0.093 (0.005) & 0.844 (0.045) \\
GCN & -- & -- & -- & -- & 0.081 (0.006) & 0.791 (0.057) \\
\bottomrule
\end{tabular}
}
\endgroup

\end{table}

\begin{table}[!t]
\centering
\caption{Estimation of $\eta$ and $H$ in the nonlinear design under SCAR over $100$ replications.}
\label{tab:eta-nonlinear-scar}
\small
\setlength{\tabcolsep}{6pt}
\begin{tabular}{lrrrrrr}
\toprule
Method
& Mean $\hat{\eta}$ & RMSE$_{\eta}$ & Cov$_{80,\eta}$ & IS$_{80,\eta}$ & Mean $\hat H$ & RMSE$_H$ \\
\midrule
GP-PU-PoE
& 0.187 & 0.149 & 0.470 & 0.441 & 6.955 & 8.171 \\
GP-PU-PoE$_{\beta_0}$
& 0.177 & 0.158 & 0.380 & 0.497 & 6.500 & 8.567 \\
\bottomrule
\end{tabular}

\end{table}

Among the competitors the ordering follows their capacity for nonlinear and relational structure. Ada and the GCN lead, at overall AUCs of $0.824$ and $0.821$ and hidden AUCs of $0.730$ and $0.687$. WLR and INLA sit close to the linear Bayesian models at $0.648$ and $0.641$, which is what the covariate design predicts, and nnPU is lowest at $0.565$, \textcolor{black}{which may reflect the small number of labeled positives in this design, even though it is given the true class prior $\pi_p=45/400$.} Appending the modularity eigenvectors helps almost everywhere, raising the overall AUC of the linear model by $0.026$, of WLR by $0.033$, of INLA by $0.033$ and of nnPU by $0.028$, with Ada the exception at $-0.010$. Proper scores on the hidden set reward methods that are optimistic everywhere, since the truth there is always one, and Ada attains the lowest hidden Brier score at $0.314$ while having an overall Brier at $0.138$, above the $0.100$ of a constant predictor at the base rate. The proposed model keeps a margin of $0.124$ in both overall and hidden AUC over the best competitor, and places $39.3\%$ of the hidden positives in the top $15$ unlabeled positions against $25.4\%$. The PU correction also improves the hidden coverage of every pair, from $0.165$ to $0.298$ for the linear covariate model, from $0.465$ to $0.567$ for the covariate expert and from $0.521$ to $0.640$ for the product of experts. Dashes mark the methods that return point estimates only and support no uncertainty quantification.

\textcolor{black}{Tables~\ref{tab:sim-sar} and~\ref{tab:eta-sar} report the results for the nonlinear design under SAR with $\lambda=2$, using the covariate dependent hiding mechanism described in Appendix~\ref{supp:hiding}. Under SAR most methods lose ground on the hidden set, since the positives that go unrecorded are by construction the ones whose covariates resemble the zeros most closely, but the proposed model stays robust to the violation and does so through the network. On covariates alone it falls from $0.763$ to $0.578$ in hidden AUC, as expected since the violation is defined entirely on the covariates. Once the network expert is present the fall stops at $0.854$ to $0.721$, so a source untouched by the mechanism mitigates part of the loss. The model also remains the best of those compared, at $0.906$ against $0.821$ for the strongest competitor on the overall set, although the margin on the hidden set narrows to $0.019$ over the GCN, whose relational input is likewise unaffected. Table~\ref{tab:eta-sar} shows the effect on the count. The estimated number of undetected cases falls from $6.9$ under SCAR to $4.5$, and the coverage of $\eta$ to $0.020$.}

\begin{table}[!t]
\centering
\caption{Nonlinear covariate-driven design under SAR with $\lambda=2$. Monte Carlo means over $100$ replications, standard deviations in parentheses. Best value in bold, closest to nominal for coverage.}
\label{tab:sim-sar}
\begingroup
\scriptsize
\setlength{\tabcolsep}{1.4pt}
\renewcommand{\arraystretch}{1.10}
\setlength{\aboverulesep}{0.15ex}
\setlength{\belowrulesep}{0.15ex}

\resizebox{\linewidth}{!}{
\begin{tabular}{lcccccc}
\toprule
Method
& \shortstack{Overall\\AUC}
& \shortstack{Hidden\\AUC}
& \shortstack{Overall\\Recall@45}
& \shortstack{Hidden\\Recall@15}
& \shortstack{Overall\\LogLoss}
& \shortstack{Hidden\\LogLoss} \\
\midrule
Lin-Cov & 0.667 (0.038) & 0.519 (0.068) & 0.301 (0.057) & 0.061 (0.059) & 0.341 (0.011) & 2.797 (0.171) \\
Lin-PU-Cov & 0.668 (0.038) & 0.519 (0.068) & 0.297 (0.059) & 0.063 (0.055) & 0.334 (0.010) & 2.604 (0.175) \\
Lin-Cov+Net & 0.691 (0.049) & 0.558 (0.090) & 0.324 (0.069) & 0.088 (0.083) & 0.334 (0.017) & 2.736 (0.191) \\
Lin-PU-Cov+Net & 0.692 (0.050) & 0.560 (0.090) & 0.319 (0.073) & 0.091 (0.085) & 0.326 (0.018) & 2.527 (0.201) \\
\hdashline
GP-Cov & 0.856 (0.024) & 0.579 (0.073) & 0.699 (0.033) & 0.157 (0.081) & 0.215 (0.011) & 3.494 (0.378) \\
GP-PU-Cov & 0.855 (0.024) & 0.578 (0.074) & 0.697 (0.032) & 0.153 (0.079) & 0.211 (0.012) & 3.305 (0.386) \\
GP-PoE & \textbf{0.906} (0.031) & 0.720 (0.095) & 0.736 (0.034) & 0.225 (0.109) & 0.190 (0.013) & 3.063 (0.423) \\
GP-PU-PoE & \textbf{0.906} (0.031) & \textbf{0.721} (0.094) & 0.736 (0.035) & 0.221 (0.107) & \textbf{0.185} (0.013) & 2.845 (0.437) \\
GP-PU-PoE$_{\beta_0}$ & 0.905 (0.031) & 0.719 (0.094) & \textbf{0.737} (0.036) & \textbf{0.228} (0.111) & 0.187 (0.012) & 2.836 (0.425) \\
\hdashline
Ada & 0.813 (0.031) & 0.570 (0.075) & 0.569 (0.048) & 0.114 (0.074) & 0.396 (0.033) & 1.421 (0.250) \\
Ada+Net & 0.814 (0.035) & 0.582 (0.093) & 0.562 (0.055) & 0.114 (0.090) & 0.428 (0.040) & \textbf{1.258} (0.250) \\
nnPU & 0.581 (0.072) & 0.486 (0.072) & 0.373 (0.066) & 0.088 (0.058) & 0.519 (0.078) & 5.568 (0.566) \\
nnPU+Net & 0.619 (0.093) & 0.521 (0.086) & 0.405 (0.092) & 0.099 (0.078) & 0.477 (0.101) & 5.286 (0.705) \\
WLR & 0.660 (0.038) & 0.515 (0.066) & 0.291 (0.063) & 0.063 (0.056) & 0.337 (0.011) & 2.352 (0.252) \\
WLR+Net & 0.688 (0.049) & 0.553 (0.084) & 0.318 (0.072) & 0.079 (0.076) & 0.328 (0.018) & 2.309 (0.268) \\
INLA & 0.654 (0.039) & 0.512 (0.067) & 0.298 (0.063) & 0.061 (0.055) & 0.353 (0.010) & 3.175 (0.290) \\
INLA+Net & 0.682 (0.049) & 0.550 (0.084) & 0.324 (0.070) & 0.079 (0.073) & 0.346 (0.017) & 3.158 (0.320) \\
GCN & 0.821 (0.037) & 0.702 (0.078) & 0.480 (0.066) & 0.197 (0.097) & 0.286 (0.023) & 2.760 (0.470) \\
\bottomrule
\end{tabular}
}

\vspace{0.3ex}

\resizebox{\linewidth}{!}{
\begin{tabular}{lcccccc}
\toprule
Method
& \shortstack{Overall\\Cov$_{80}$}
& \shortstack{Hidden\\Cov$_{80}$}
& \shortstack{Overall\\IS$_{80}$}
& \shortstack{Hidden\\IS$_{80}$}
& \shortstack{Overall\\Brier}
& \shortstack{Hidden\\Brier} \\
\midrule
Lin-Cov & \textbf{0.926} (0.015) & 0.149 (0.102) & 0.889 (0.095) & 8.662 (0.914) & 0.096 (0.004) & 0.865 (0.016) \\
Lin-PU-Cov & 0.942 (0.012) & 0.268 (0.121) & 0.844 (0.078) & 7.588 (1.085) & 0.095 (0.004) & 0.836 (0.019) \\
Lin-Cov+Net & 0.931 (0.015) & 0.203 (0.131) & 0.848 (0.105) & 8.176 (1.181) & 0.095 (0.005) & 0.851 (0.033) \\
Lin-PU-Cov+Net & 0.947 (0.011) & 0.327 (0.136) & 0.804 (0.091) & 7.059 (1.228) & 0.093 (0.005) & 0.817 (0.043) \\
\hdashline
GP-Cov & 0.968 (0.004) & 0.203 (0.097) & 0.451 (0.035) & 8.170 (0.871) & 0.053 (0.004) & 0.882 (0.035) \\
GP-PU-Cov & 0.970 (0.004) & 0.231 (0.103) & 0.463 (0.035) & 7.918 (0.924) & 0.052 (0.004) & 0.859 (0.042) \\
GP-PoE & 0.972 (0.005) & 0.269 (0.136) & \textbf{0.403} (0.043) & 7.576 (1.227) & 0.049 (0.003) & 0.853 (0.045) \\
GP-PU-PoE & 0.975 (0.005) & 0.335 (0.151) & 0.406 (0.047) & 6.988 (1.359) & \textbf{0.048} (0.003) & 0.821 (0.057) \\
GP-PU-PoE$_{\beta_0}$ & 0.975 (0.006) & 0.333 (0.156) & 0.406 (0.048) & 7.000 (1.401) & 0.049 (0.003) & 0.823 (0.054) \\
\hdashline
Ada & -- & -- & -- & -- & 0.121 (0.013) & 0.506 (0.074) \\
Ada+Net & -- & -- & -- & -- & 0.134 (0.016) & \textbf{0.459} (0.083) \\
nnPU & -- & -- & -- & -- & 0.078 (0.008) & 0.940 (0.033) \\
nnPU+Net & -- & -- & -- & -- & 0.074 (0.010) & 0.934 (0.041) \\
WLR & 0.973 (0.007) & 0.589 (0.152) & 0.864 (0.071) & 4.702 (1.365) & 0.095 (0.004) & 0.766 (0.036) \\
WLR+Net & 0.974 (0.006) & \textbf{0.607} (0.143) & 0.828 (0.088) & \textbf{4.533} (1.291) & 0.093 (0.006) & 0.747 (0.054) \\
INLA & 0.938 (0.006) & 0.225 (0.096) & 0.841 (0.058) & 7.978 (0.868) & 0.094 (0.004) & 0.879 (0.025) \\
INLA+Net & 0.941 (0.007) & 0.259 (0.115) & 0.812 (0.079) & 7.665 (1.038) & 0.093 (0.006) & 0.866 (0.041) \\
GCN & -- & -- & -- & -- & 0.081 (0.006) & 0.784 (0.045) \\
\bottomrule
\end{tabular}
}
\endgroup

\end{table}

\begin{table}[!t]
\centering
\caption{Estimation of $\eta$ and $H$ in the nonlinear design under SAR with $\lambda=2$ over $100$ replications.}
\label{tab:eta-sar}
\small
\setlength{\tabcolsep}{6pt}
\begin{tabular}{lrrrrrr}
\toprule
Method
& Mean $\hat{\eta}$ & RMSE$_{\eta}$ & Cov$_{80,\eta}$ & IS$_{80,\eta}$ & Mean $\hat H$ & RMSE$_H$ \\
\midrule
GP-PU-PoE
& 0.130 & 0.205 & 0.020 & 1.139 & 4.479 & 10.549 \\
GP-PU-PoE$_{\beta_0}$
& 0.127 & 0.208 & 0.000 & 1.184 & 4.358 & 10.664 \\
\bottomrule
\end{tabular}

\end{table}

\textcolor{black}{In the linear design under SCAR, the model remains competitive despite the simpler data-generating mechanism, as shown in Appendix~\ref{supp:sim}.}

\section{Detecting fraud in the Greek \textcolor{black}{energy market}}
\label{sec:appl}

\subsection{Data and prior setting} \label{sec:data}

We analyze administrative data on \textcolor{black}{a regulated segment of the Greek energy market} for $2023$, provided by the Independent Authority for Public Revenue (IAPR), the Greek tax authority. The records are confidential and were made available in anonymized form, with all identifiers removed before analysis. \textcolor{black}{The fitted sample contains $n=550$ firms prioritized for inspection during the period. As a result, the analysis concerns risk differentiation among the firms prioritized for inspection rather than prevalence among all firms in this market segment.} Of these, $37$ carry a positive label, meaning that an inspection established non-compliance, while the rest carry a zero, meaning either that no inspection took place or that one took place and found nothing. The PU likelihood allows this zero-labeled group to contain undetected positives rather than treating all such \textcolor{black}{firms} as compliant.

\textcolor{black}{Each firm has $p=13$ covariates. Three describe commercial activity and pricing, and the remaining ten are yearly counts of alerts raised by the electronic monitoring system that every firm is required to operate, each category recording a distinct type of abnormal operation.} The covariates are transformed by $\log(x+1)$ and standardized on the fitted sample. \textcolor{black}{Three network layers describe the relational structure among the firms. For confidentiality reasons, we refer to them as Network~1, Network~2 and Network~3 without disclosing their specific definitions.}

Priors follow the elicitation of Appendix~\ref{supp:prior-elicitation}, calibrated on the observed designs. We set the target standard deviation of the total stochastic latent field to $c=2$ and give the covariate expert half of the prior variance, the three networks sharing the other half equally, so that $w=(1/2,1/6,1/6,1/6)$. The lengthscale quantiles of each expert are matched at the $10\%$ and $50\%$ points of the correlation range its own kernel can attain, the Half-$t_3$ scales place the structured mean of each expert on its own scale $c_j$, the
baseline risk is given a central $95\%$ range of $1\%$ to $10\%$, and the false negative rate a prior $95$th percentile of one half. \textcolor{black}{This value reflects operational information from the IAPR, which indicates that some noncompliant firms escape detection but that detection succeeds more often than it fails when noncompliance is present.} The elicited hyperparameters are in Appendix~\ref{supp:elicited-priors}.  Regarding the network experts, the modularity eigengap rule keeps \textcolor{black}{five eigenvectors for Network~1, two for Network~2 and one for Network~3} in the structured means. 

\subsection{Validation} \label{sub:val_appl}
 
Since only the $37$ observed positives have verified class status, predictive ranking is evaluated through leave-one-held-out-positive cross-validation (LOO-HOP).

In each fold one observed positive is relabeled as zero and the model is refitted following Algorithm~\ref{alg:schedule}, with $15{,}000$ warm-up iterations and $40{,}000$ sampling iterations thinned by $20$, retaining $2{,}000$ draws. The GP models are fitted exactly, without a sparse approximation. The warm-up of the full fit supplies the starting values of every fold, which helps each chain settle quickly. We record where the held out \textcolor{black}{firm} ranks among the unlabeled ones, and across the $37$ folds we report the mean AUC, the recall within the top $37$ and the top $50$ positions, the mean rank and the mean reciprocal rank (MRR).

Table~\ref{tab:loo-pos} presents the results. Lin-Cov reaches an AUC
of $0.904$, Lin-PU-Cov $0.916$, Lin-Cov+Net $0.941$ and
Lin-PU-Cov+Net $0.957$. Moving to the GP models gives an AUC of
$0.964$ for GP-Cov and $0.975$ for GP-PoE, with the latter recovering
$97.3\%$ of the held out positives within the top $37$ positions.
Each competing method generally improves once network information is
included, most notably nnPU, which moves from $0.767$ to $0.968$ with
nnPU+Net. The improvement of the network augmented competitors shows
that the relational signal is not specific to the PoE construction.
GP-PU-PoE reaches an AUC of $0.974$, against $0.968$ for nnPU+Net,
$0.951$ for Ada+Net and $0.948$ for the GCN. GP-PoE,
GP-PU-PoE and GP-PU-PoE$_{\beta_0}$ remain very close in ranking
performance, with AUCs of $0.975$, $0.974$ and $0.979$, respectively. Their larger differences lie in the fitted probabilities and in the
estimated number of undetected cases rather than in the rankings, and
Appendix~\ref{supp:mu0} compares the last two in full.

\begin{table}[t]
\centering
\caption{LOO-HOP evaluation over the $37$ observed positives, reporting point estimates with $95\%$ bootstrap intervals in brackets.}
\label{tab:loo-pos}
\footnotesize
\setlength{\tabcolsep}{3.5pt}
\resizebox{\linewidth}{!}{
\begin{tabular}{lccccc}
\toprule
Method & LOO-HOP AUC & Recall@37 & Recall@50 & Mean rank & MRR \\
\midrule
Lin-Cov & 0.904 [0.858, 0.942] & 0.676 [0.541, 0.811] & 0.730 [0.595, 0.865] & 50.0 [30.8, 73.7] & 0.101 [0.056, 0.165] \\
Lin-PU-Cov & 0.916 [0.872, 0.952] & 0.676 [0.514, 0.811] & 0.730 [0.595, 0.865] & 43.9 [25.6, 66.7] & 0.120 [0.071, 0.186] \\
Lin-Cov+Net & 0.941 [0.902, 0.967] & 0.784 [0.649, 0.919] & 0.865 [0.757, 0.973] & 31.4 [18.1, 51.4] & 0.106 [0.064, 0.168] \\
Lin-PU-Cov+Net & 0.957 [0.934, 0.972] & 0.865 [0.757, 0.973] & 0.919 [0.811, 1.000] & 23.1 [15.2, 34.9] & 0.108 [0.067, 0.170] \\
\hdashline
GP-Cov & 0.964 [0.933, 0.983] & 0.919 [0.811, 1.000] & \textbf{0.973 [0.919, 1.000]} & 19.27 [9.59, 35.41] & 0.288 [0.184, 0.403] \\
GP-PU-Cov & 0.963 [0.934, 0.981] & 0.865 [0.757, 0.973] & 0.946 [0.865, 1.000] & 19.95 [10.59, 34.97] & 0.234 [0.147, 0.332] \\
GP-PoE & 0.975 [0.956, 0.987] & \textbf{0.973 [0.919, 1.000]} & \textbf{0.973 [0.919, 1.000]} & 13.70 [7.70, 23.41] & 0.285 [0.190, 0.396] \\
GP-PU-PoE & 0.974 [0.957, 0.986] & 0.946 [0.865, 1.000] & \textbf{0.973 [0.919, 1.000]} & 14.14 [8.41, 23.14] & 0.218 [0.142, 0.307] \\
\textbf{GP-PU-PoE$_{\beta_0}$} & \textbf{0.979 [0.963, 0.989]} & \textbf{0.973 [0.919, 1.000]} & \textbf{0.973 [0.919, 1.000]} & \textbf{12.03 [6.89, 20.11]} & \textbf{0.306 [0.214, 0.405]} \\
\hdashline
Ada & 0.943 [0.905, 0.968] & 0.703 [0.541, 0.838] & 0.892 [0.784, 0.973] & 30.2 [17.2, 49.7] & 0.176 [0.099, 0.272] \\
Ada+Net & 0.951 [0.918, 0.973] & 0.892 [0.784, 0.973] & 0.919 [0.811, 1.000] & 25.9 [15.1, 43.1] & 0.121 [0.081, 0.168] \\
nnPU & 0.767 [0.673, 0.849] & 0.459 [0.297, 0.622] & 0.514 [0.351, 0.676] & 120.8 [78.3, 168.6] & 0.083 [0.048, 0.124] \\
nnPU+Net & 0.968 [0.948, 0.984] & 0.892 [0.784, 0.973] & 0.892 [0.784, 0.973] & 17.2 [9.1, 27.8] & 0.237 [0.162, 0.327] \\
WLR & 0.905 [0.864, 0.940] & 0.622 [0.459, 0.784] & 0.703 [0.541, 0.838] & 49.9 [31.7, 70.8] & 0.087 [0.059, 0.117] \\
WLR+Net & 0.928 [0.869, 0.965] & 0.730 [0.595, 0.865] & 0.838 [0.703, 0.946] & 37.7 [18.7, 68.4] & 0.183 [0.096, 0.286] \\
INLA & 0.889 [0.846, 0.928] & 0.595 [0.432, 0.757] & 0.676 [0.541, 0.811] & 57.7 [37.8, 80.1] & 0.077 [0.051, 0.106] \\
INLA+Net & 0.911 [0.848, 0.956] & 0.730 [0.595, 0.865] & 0.757 [0.622, 0.892] & 46.4 [23.6, 79.1] & 0.146 [0.079, 0.227] \\
GCN & 0.948 [0.897, 0.977] & 0.946 [0.865, 1.000] & 0.946 [0.865, 1.000] & 27.7 [12.7, 54.0] & 0.204 [0.117, 0.307] \\
\bottomrule
\end{tabular}
}

\end{table}

\subsection{Posterior risk and audit implications}
\label{sub:audit}

In this subsection we focus on the full GP-PU-PoE model, fitted on all $550$ \textcolor{black}{firms} following again Algorithm~\ref{alg:schedule} with $60{,}000$ warm-up iterations, the first third blocked by expert and the rest joint, then $250{,}000$ sampling iterations thinned by $50$, retaining $5{,}000$ draws. Diagnostics and posterior summaries are in
Appendix~\ref{supp:posterior-diagnostics}. The posterior mean of $\eta$ is $0.137$, with standard deviation $0.076$ and a $5\%$ to
$95\%$ interval of $[0.029,0.277]$. Conditional on the SCAR model and the identification assumptions of Section~\ref{sec:eta}, this suggests that most fraudulent \textcolor{black}{firms} in the fitted sample are recorded, although inference on the nondetection rate is substantially more uncertain than the risk ranking.

\textcolor{black}{Summing the individual probabilities of being an undetected positive gives $\mathbb E(H\mid\mathbf y)=5.25$ in \eqref{eq:H-posterior}. Drawing the latent indicators gives an integer valued posterior with median $5$ and an $80\%$ HPD interval of $[0,8]$. Under SCAR, the estimate based on $\eta$ gives $\widehat H_\eta=6.26$, with an $80\%$ HPD interval of $[0.66,9.64]$.} Figure~\ref{fig:pupoe-pred} shows the posterior probabilities with $80\%$ HPD intervals in the top panel and the four expert contributions $\sigma(x_{j,i})$ below. The separation is clear, with a small number of zero labeled \textcolor{black}{firms} carrying a raised probability and five of them above one half. The lower panels show that all four experts contribute, with \textcolor{black}{Network~1} giving the cleanest separation between the two groups.

Table~\ref{tab:application-candidates} records the six candidates with anonymized identifiers. The networks carry the signal, and the installer layer most of it. The full decomposition of the expert shares is reported in
Appendix~\ref{supp:expert-shares}.

\textcolor{black}{The estimate based on $\eta$ under SCAR suggests about six undetected cases, and given the capacity of the IAPR for follow-up inspections, we focused on the six highest ranked zero labeled firms. When we discussed the results with the IAPR, two of them had already been reinspected independently of our analysis, before the results were communicated, and both had been confirmed as noncompliant. This provides an independent retrospective validation of these two cases. The remaining four had not yet been reinspected and were selected by the IAPR for follow-up inspection. For two of these four, the covariate expert points away from noncompliance, while their elevated risk is driven by the network experts. Specifically, the first and the fourth candidate show most clearly what the relational structure adds, with their covariate expert values of $0.380$ and $0.364$. In both cases the evidence comes from the relational side of the model, which is consistent with the possibility that noncompliant firms misreport the figures they submit. The same pattern holds across all firms, as the posterior share of the covariate expert falls from its prior value of $0.500$ to $0.301$, while the combined share of the three network experts rises from $0.500$ to $0.699$.}

\begin{figure}[tbp]
  \centering
  \includegraphics[width=0.875\linewidth]
    {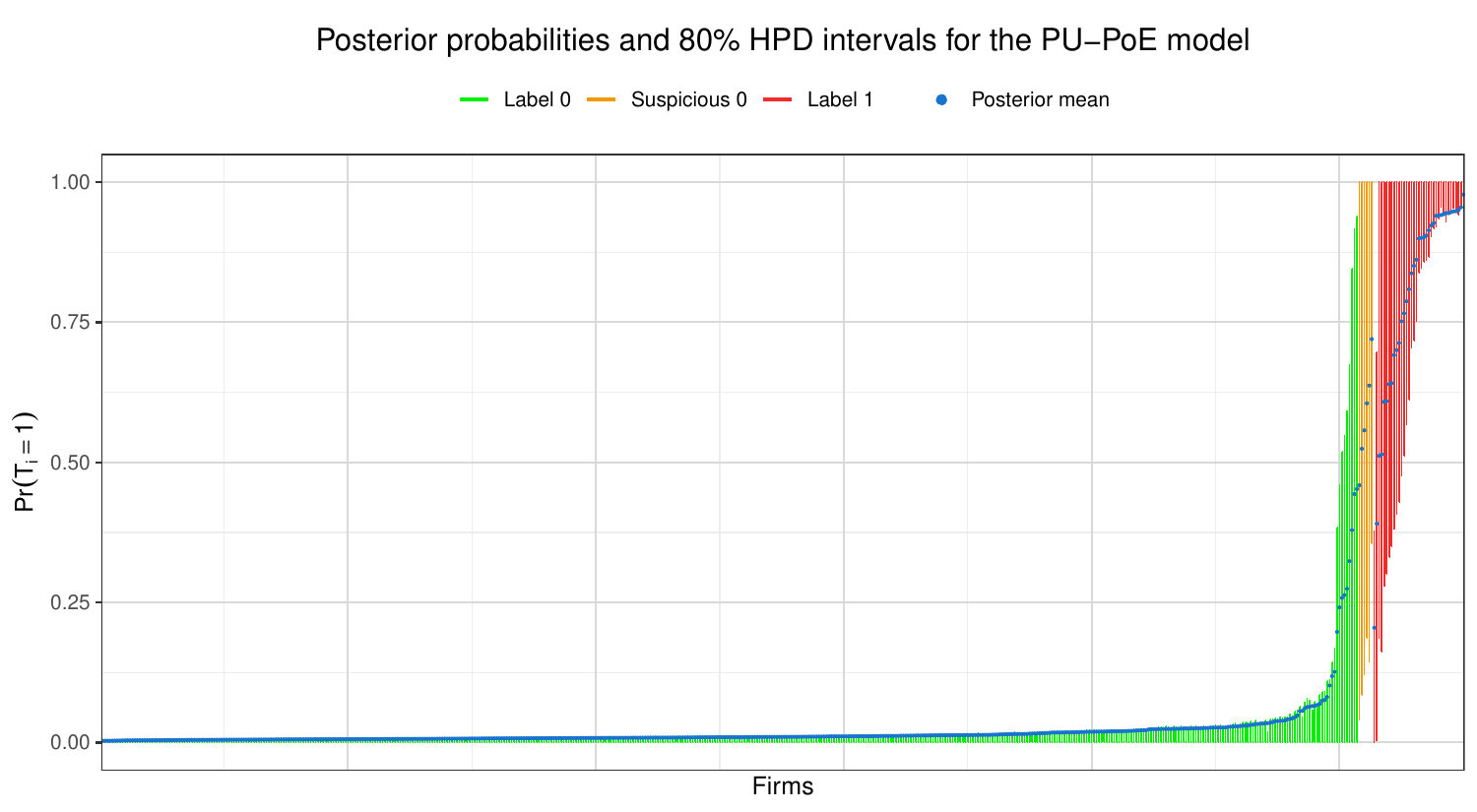}
  \vspace{0.5em}
  \includegraphics[width=0.875\linewidth]
    {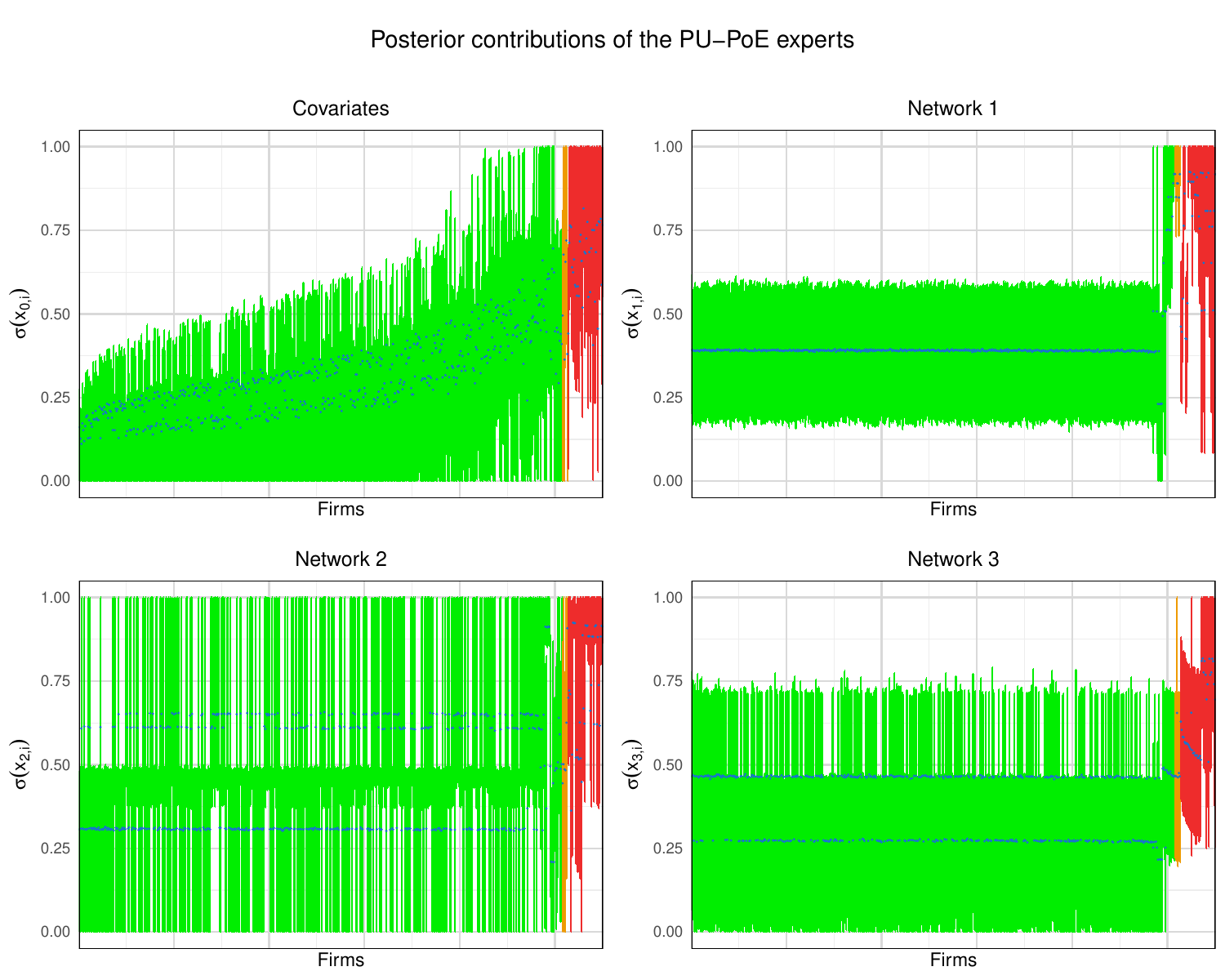}
    
  \caption{Top, posterior probabilities and $80\%$ HPD intervals.
  Bottom, expert specific contributions $\sigma(x_{j,i})$.}
  \label{fig:pupoe-pred}
\end{figure}

\begin{table}[tbp]
\centering
\caption{The six screened \textcolor{black}{firms}, with overall and expert specific posterior estimates.}
\label{tab:application-candidates}
\small
\renewcommand{\arraystretch}{1.08}
\resizebox{0.8\linewidth}{!}{
\begin{tabular}{lrrrrr}
\toprule
\multirow{2}{*}{Candidate} &
Overall &
Covariates &
Network 1 &
Network 2 &
Network 3 \\[-0.6ex]
&
$\hat p_i$ &
$\sigma(x_{0,i})$ &
$\sigma(x_{1,i})$ &
$\sigma(x_{2,i})$ &
$\sigma(x_{3,i})$ \\
\midrule
C1 & 0.720 & 0.380 & 0.918 & 0.908 & 0.465 \\
C2 & 0.637 & 0.720 & 0.848 & 0.626 & 0.475 \\
C3 & 0.605 & 0.656 & 0.890 & 0.484 & 0.462 \\
C4 & 0.557 & 0.364 & 0.890 & 0.486 & 0.655 \\
C5 & 0.524 & 0.677 & 0.849 & 0.502 & 0.465 \\
C6 & 0.459 & 0.495 & 0.918 & 0.445 & 0.467 \\
\bottomrule
\end{tabular}
}
\end{table}

\section{Discussion and Limitations}
\label{sec:disc}

We developed a Bayesian framework for fraud detection under positive unlabeled labels that combines covariates with multilayer network information. Each source contributes a GP expert with a structured mean and its own covariance kernel, and the model estimates the false negative rate, posterior fraud probabilities, the number of undetected cases, and source specific contributions within a single posterior analysis. Under an anchor condition the false negative rate is bounded. For a fixed latent risk function, a common false negative rate leaves the induced ranking unchanged. Computationally, once the aggregate covariance is factorized, the cost of a latent field update does not increase with the number of experts, although hyperparameter proposals require a new factorization.

The application shows that the model components serve different purposes and should not be judged by a single metric. The network experts mainly improve ranking, the PU component mainly affects calibration and inference about undetected cases, and the structured mean sharpens confidence in observed positives, although no individual coefficient separates clearly from zero. \textcolor{black}{Of the six highest ranked firms, two were confirmed as noncompliant through independent reinspection, and the remaining four, including the two whose elevated risk comes mainly from the network experts, were selected by the IAPR for follow-up inspection.} For an audit authority these are distinct products, a better ordered inspection list, an estimate of how much fraud current controls leave undetected, and a statement of which relationships drove each selection. \textcolor{black}{These relationships provide predictive information about risk, but should not be interpreted as evidence that a connected party was involved in the violation.}

Several limitations remain. The likelihood assumes a common false negative rate, whereas audit selection may depend on covariates, geography, or prior intelligence. Our covariate dependent simulation indicates some robustness to this form of misspecification, but an explicit mechanism, with a unit specific rate driven by observed audit effort, would be preferable. Identification also relies on the anchor condition, that some firms carry near certain fraud risk. This is plausible in audit data but cannot be verified, because the labels record only how often fraud is detected, which combines how often it occurs with how often audits catch it. Consistent with this limitation, the simulations show that the false-negative rate and the number of hidden positives can be difficult to estimate accurately even when risk ranking is strong. The dense implementation targets moderate sample sizes, so sparse or low rank representations are needed at national scale. Finally, prior variance is allocated across experts rather than learned, although their posterior contributions can adapt to the data. Strong dependence across experts can produce overconfident predictions, and a generalized product of experts that downweights redundant experts \citep{cao2014generalized}, or an explicit expert selection prior, is a natural extension.

Overall, the framework offers a coherent Bayesian tool for risk based auditing with heterogeneous data, and more broadly a template for settings where confirmed cases are few, unlabeled units cannot be treated as negatives, and relational information is harder to manipulate than reported figures.

\section*{Software and Reproducibility}
\phantomsection
\addcontentsline{toc}{section}{Software and Reproducibility}
\textcolor{black}{The \textsf{R} code for the illustrative example of Section~\ref{sec:illustration}, its expert decomposition, the prior sensitivity analysis, and the three simulation designs is available at \url{https://github.com/KonstantinosBourazas/GP_PU-PoE}. The administrative data analyzed in Section~\ref{sec:appl} are confidential, as they contain protected tax records and commercially sensitive information, and cannot be made publicly available. Thus, the application analysis, including the LOO-HOP evaluation, is not part of the repository.}

\FloatBarrier

\section*{Acknowledgments}
\phantomsection
\addcontentsline{toc}{section}{Acknowledgments}
We are grateful to the Independent Authority for Public Revenue of the Hellenic Republic (IAPR) for supporting this research and to the many officials who have provided feedback through extensive discussions. We also thank participants at the Workshop on ``New Perspectives on Tax Fraud Detection'', December 1, 2025, Athens, Greece, for helpful comments and suggestions. \textcolor{black}{The methodology in this paper was developed by the authors, and the views expressed are theirs and do not necessarily reflect those of the IAPR and its Management.}

\section*{Funding}
\phantomsection
\addcontentsline{toc}{section}{Funding}
Alexopoulos and Bourazas gratefully acknowledge funding from the Hellenic Foundation for Research and Innovation (H.F.R.I.) under the National Recovery and Resilience Plan ``Greece 2.0'', funded by the European Union—NextGenerationEU (H.F.R.I. Project No.\ 15973).

\FloatBarrier
\bibliographystyle{plainnat}
\bibliography{refs}

\clearpage
\appendix
\titleformat{\section}{\large\bfseries}{Appendix \thesection.}{0.6em}{}
\numberwithin{equation}{section}
\numberwithin{figure}{section}
\numberwithin{table}{section}
\numberwithin{algorithm}{section}
\renewcommand{\theHequation}{\thesection.\arabic{equation}}
\renewcommand{\theHfigure}{\thesection.\arabic{figure}}
\renewcommand{\theHtable}{\thesection.\arabic{table}}
\providecommand{\theHalgorithm}{}
\renewcommand{\theHalgorithm}{\thesection.\arabic{algorithm}}

\section{Proofs}
\label{supp:proofs}

\subsection{Proof of Remark~\ref{lemma:poe-logistic}}

\begin{proof}
For $p_{j,i}=\sigma(x_{j,i})$,
$p_{j,i}/(1-p_{j,i})=\exp(x_{j,i})$. Dividing the numerator and denominator of \eqref{eq:poe-prob} by $\prod_{j=0}^m(1-p_{j,i})$ gives
\[
\Pr(T_i=1\mid x_{0:m,i},\beta_0)
=
\frac{\exp\{\beta_0+\sum_{j=0}^m x_{j,i}\}}
{1+\exp\{\beta_0+\sum_{j=0}^m x_{j,i}\}}
=
\sigma(f_i).
\]
\end{proof}

\subsection{Proof of Lemma~\ref{lem:ID-ROC}}

\begin{proof}
For part (a), let
$r^\star=\operatorname*{ess\,sup}_x\sigma(f(x))$ and
$q^\star=\operatorname*{ess\,sup}_xq(x)$. Since
$q^\star=(1-\eta)r^\star$ and
$r^\star\in[1-\delta,1]$,
\[
q^\star
\le
1-\eta
\le
\frac{q^\star}{1-\delta}.
\]
Subtracting from one and imposing $\eta\ge0$ yields \eqref{eq:eta-bounds}. The interval has width at most
$q^\star\delta/(1-\delta)$, and when $\delta=0$ it collapses to
$\eta=1-q^\star$, since $r^\star=1$ and $q^\star=1-\eta$. Therefore,
$\sigma(f(x))=q(x)/(1-\eta)$ $P_X$-almost everywhere, and $f$ is
recovered by the inverse logit transformation.

For part (b), since $1-\eta>0$, for every pair $i,j$,
\[
\sigma(f_i)\leq\sigma(f_j)
\quad\Longleftrightarrow\quad
(1-\eta)\sigma(f_i)\leq(1-\eta)\sigma(f_j).
\]
Thus multiplication by the common positive factor $1-\eta$
preserves the ordering of the scores. The two score vectors
therefore induce the same ROC curve and AUC.

\end{proof}

\subsection{Proof of Remark~\ref{rem:eta-concave}}

\begin{proof}
Conditionally on $\mathbf f$, the log posterior of $\eta$, up to
terms that do not depend on $\eta$, is
\[
\begin{aligned}
\log\pi(\eta\mid\mathbf f,\mathbf y)
&=
\sum_{i:y_i=1}\log(1-\eta)
+
\sum_{i:y_i=0}\log(1-p_i+\eta p_i) \\
&\quad +
(a_\eta-1)\log\eta
+
(b_\eta-1)\log(1-\eta).
\end{aligned}
\]
Differentiating twice gives
\[
\frac{\partial^2}{\partial\eta^2}
\log\pi(\eta\mid\mathbf f,\mathbf y)
=
-\sum_{i:y_i=1}\frac{1}{(1-\eta)^2}
-\sum_{i:y_i=0}
\frac{p_i^2}{(1-p_i+\eta p_i)^2}
-\frac{a_\eta-1}{\eta^2}
-\frac{b_\eta-1}{(1-\eta)^2}.
\]
For finite $f_i$, $p_i\in(0,1)$, so every likelihood contribution
is strictly negative. The two prior contributions are nonpositive
when $a_\eta,b_\eta\ge1$. Since $n\ge1$, the second derivative is
strictly negative on $(0,1)$, proving strict concavity.
\end{proof}

\subsection{Proof of Remark~\ref{rem:geometry-f}}

\begin{proof}
The log likelihood in \eqref{eq:pu-loglik} is separable, with
$\mathcal L(\mathbf f;\eta)=\sum_{i=1}^n\mathcal L_i(f_i;\eta)$, where
\[
\mathcal L_i(f_i;\eta)
=
y_i\log\big((1-\eta)p_i\big)
+
(1-y_i)\log\big(1-(1-\eta)p_i\big),
\qquad
p_i=\sigma(f_i).
\]
Since $\partial p_i/\partial f_i=p_i(1-p_i)$,
\[
\frac{\partial\mathcal L_i(f_i;\eta)}{\partial f_i}
=
y_i(1-p_i)
-
(1-y_i)
\frac{(1-\eta)p_i(1-p_i)}
{1-(1-\eta)p_i},
\]
which gives \eqref{eq:pu-gradient}. When $\eta=0$, this reduces to
$y_i(1-p_i)-(1-y_i)p_i=y_i-p_i$.

For $y_i=1$,
$\partial^2\mathcal L_i/\partial f_i^2=-p_i(1-p_i)<0$.
For $y_i=0$, write $A=1-\eta$, $p=\sigma(f_i)$,
$U=p(1-p)$, and $D=1-Ap$. Since
$\partial U/\partial f_i=U(1-2p)$ and
$\partial D/\partial f_i=-AU$, direct differentiation gives
\[
\frac{\partial^2\mathcal L_i(f_i;\eta)}{\partial f_i^2}
=
\frac{-AU(1-2p)D-A^2U^2}{D^2}
=
\frac{Ap(1-p)\{2p-1-Ap^2\}}{(1-Ap)^2}.
\]
As $f_i\to\infty$, $p=\sigma(f_i)\to1$, and hence
$2p-1-Ap^2\to1-A=\eta>0$. Thus, for any zero labeled observation, the
corresponding likelihood term is locally convex for sufficiently
large finite $f_i$. If at least one $y_i=0$, the log likelihood
cannot therefore be globally concave in $\mathbf f$.
\end{proof}

The conditional posterior Hessian is
$\operatorname{diag}\{\mathcal L_1'',\ldots,\mathcal L_n''\}
-\mathbf C_{\mathbf f}^{-1}$.
A sufficiently informative Gaussian prior can restore log concavity,
so the remark concerns the likelihood rather than every possible
posterior specification.

\subsection{Proof of Proposition~\ref{prop:poe-cost}}

\begin{proof}
Throughout, we condition on
$(\beta_0,\boldsymbol\gamma,\boldsymbol\theta,\eta)$
and suppress this conditioning from the notation.
By prior independence,
\[
p(\mathbf x_0,\ldots,\mathbf x_m)
=
\prod_{j=0}^m
\mathcal N(\mathbf x_j;\boldsymbol\mu_j,\mathbf C_j),
\]
and we write
$\mathbf h=\mathbf f-\boldsymbol\mu_{\mathbf f}
=\sum_{j=0}^m(\mathbf x_j-\boldsymbol\mu_j)$.
Since the centered expert states are independent, for every
$\mathbf t\in\mathbb R^n$,
\[
\begin{aligned}
\mathbb E\!\left[\exp\{\mathrm{i}\mathbf t^\top\mathbf h\}\right]
&=
\prod_{j=0}^m
\mathbb E\!\left[
\exp\{\mathrm{i}\mathbf t^\top(\mathbf x_j-\boldsymbol\mu_j)\}
\right] \\
&=
\prod_{j=0}^m
\exp\!\left(
-\frac{1}{2}\mathbf t^\top\mathbf C_j\mathbf t
\right)
=
\exp\!\left(
-\frac{1}{2}\mathbf t^\top\mathbf C_{\mathbf f}\mathbf t
\right),
\end{aligned}
\]
so that summing the expert contributions induces
$\mathbf h\sim\mathcal N(\mathbf 0,\mathbf C_{\mathbf f})$
with $\mathbf C_{\mathbf f}=\sum_{j=0}^m\mathbf C_j$.

Because the likelihood depends on the expert states only through
$\mathbf f=\boldsymbol\mu_{\mathbf f}+\mathbf h$, marginalizing them
while retaining their sum gives the exact conditional posterior
\[
\pi(\mathbf h\mid
\mathbf y,\beta_0,\boldsymbol\gamma,\boldsymbol\theta,\eta)
\propto
\exp\!\left\{
\mathcal L(\boldsymbol\mu_{\mathbf f}+\mathbf h;\eta)
\right\}
\mathcal N(\mathbf h;\mathbf 0,\mathbf C_{\mathbf f}),
\]
which is of the latent Gaussian form of
\citet{titsias2018auxiliary}, with prior covariance
$\mathbf C_{\mathbf f}$ and gradient
$\nabla_{\mathbf h}
\mathcal L(\boldsymbol\mu_{\mathbf f}+\mathbf h;\eta)
=\nabla_{\mathbf f}\mathcal L(\mathbf f;\eta)$.

Given $\mathbf z$, the prior informed proposal satisfies, on
completing the square,
\[
q(\mathbf h'\mid\mathbf z)
\propto
\mathcal N(\mathbf h';\mathbf 0,\mathbf C_{\mathbf f})
\mathcal N\!\left(
\mathbf z;\mathbf h',\dfrac{\delta}{2}\mathbf I
\right)
=
Z(\mathbf z)\,
\mathcal N\!\left(
\mathbf h';
\dfrac{2}{\delta}\mathbf A\mathbf z,
\mathbf A
\right),
\]
where
$\mathbf A=(\mathbf C_{\mathbf f}^{-1}+2\mathbf I/\delta)^{-1}$
and
$Z(\mathbf z)=
\mathcal N(\mathbf z;\mathbf 0,
\mathbf C_{\mathbf f}+\delta\mathbf I/2)$,
so the proposal is that of \eqref{eq:agrad}. Given $\mathbf z$, both $Z(\mathbf z)$ and the Gaussian
factor cancel from the Metropolis--Hastings ratio, leaving the
acceptance probability \eqref{eq:agrad-acceptance}, so the update
targets the exact conditional posterior of $\mathbf h$ and hence of
$\mathbf f$. The experts enter only through $\mathbf C_{\mathbf f}$,
and once this is factorized each update costs $\mathcal O(n^2)$: a
fixed number of dense matrix vector products and one evaluation of
the likelihood and its gradient, none of them expert specific. Only
the construction and summation of the $\mathbf C_j$ depends on $m$.

\end{proof}

\section{Prior specification and the illustrative example}
\label{supp:prior-specification}

\subsection{Hyperparameter elicitation}
\label{supp:prior-elicitation}

In Gaussian process classification the kernel hyperparameters are usually selected by maximizing an approximate marginal likelihood \citep{kuss2005assessing,nickisch2008approximations}, which uses the labels both to set the prior and to fit the model, and relies on approximations centered at a posterior mode. Neither is appealing here, since the PU log likelihood is not concave and $\eta$ is only partially identified. For this reason, we calibrate the priors for the GP experts from the covariate and network designs alone, before observing any labels. We keep these priors weakly informative, as vague priors leave the flat directions of the likelihood unresolved while sharp ones would displace the evidence in the labels.
 
Although the hyperpriors for every $l_j$ and $\sigma_j$ belong to the same family, the same numerical values mean very different things across experts, since a Euclidean lengthscale and a graph lengthscale are not comparable and a spectral network kernel does not have a constant diagonal. We therefore calibrate both through summaries that each kernel induces on the observed design. Write $\mathbf C_j=\sigma_j^2\mathbf K_j(l_j)$, where $\mathbf K_j(l)$ is the kernel with the amplitude factored out.
 
For the amplitudes, let $c>0$ be the target standard deviation of the total stochastic latent field. We assign known shares $w_j\ge0$ with $\sum_{j=0}^m w_j=1$ and give expert $j$ the scale $c_j=c\sqrt{w_j}$, so that $\sum_j c_j^2=c^2$. A larger $w_j$ gives expert $j$ more prior influence, while a larger $c$ makes the prior more aggressive overall, pushing the fraud probabilities toward zero and one. The prior for $\sigma_j$ is calibrated at the prior median lengthscale $\exp(\mu_{l,j})$, so that the effective marginal standard deviation there has central $95\%$ interval $[c_j/2,\,2c_j]$. For the covariate expert $\mathbf K_0$ has unit diagonal, so this standard deviation is $\sigma_0$ itself. For a network expert it is $\sigma_jd_j$, where $d_j(l)=n^{-1}\sum_{i=1}^n[\mathbf K_j(l)]_{ii}^{1/2}$ averages the marginal standard deviations of the kernel once the amplitude is factored out.
 
For the lengthscales, define the average induced correlation
\[
\bar\rho_j(l)
=
\binom n2^{-1}\sum_{i<i'}
\frac{[\mathbf K_j(l)]_{ii'}}
{[\mathbf K_j(l)]_{ii}^{1/2}[\mathbf K_j(l)]_{i'i'}^{1/2}}.
\]
Higher average correlation makes the latent field vary smoothly and globally, while lower correlation makes it local. For the covariate expert $\bar\rho_0(l)$ sweeps the whole of $(0,1)$ as $l$ varies, but on a network the attainable range can be far narrower, for instance when many nodes are isolated, so absolute correlation targets are not transferable. We compute the attainable range $[\rho_{j,\min},\rho_{j,\max}]$ over a wide grid of lengthscales, take the values of $l$ attaining two prescribed fractions of it, and match them to the $2.5\%$ and $97.5\%$ quantiles of the prior for $l_j$. We use the fractions $10\%$ and $50\%$. The two inverted lengthscales are ordered before matching, since the correlation curve rises with $l$ for the covariate kernel and falls with it on a network.
 
For a network expert the correlation curve is evaluated separately on each connected component, retaining the zero Laplacian eigenvalue. A clique of size $s$, the typical component in our layers, needs no eigendecomposition, since
\[
\mathbf K(l)=q(l)\mathbf I_s+\frac{1-q(l)}{s}\mathbf1_s\mathbf1_s^\top,
\qquad
q(l)=\exp\!\Big\{-\frac{s/(s-1)}{2l^2}\Big\}.
\]
 
Both the lengthscale and the amplitude priors are then obtained by fitting a log-normal to two positive quantiles. Matching $a<b$ at central probabilities $p_{\mathrm{lo}},p_{\mathrm{hi}}$ gives
\[
s=\frac{\log b-\log a}{z_{p_{\mathrm{hi}}}-z_{p_{\mathrm{lo}}}},
\qquad
\mu=\log a-z_{p_{\mathrm{lo}}}s,
\]
applied to the inverted lengthscales and to the amplitude interval divided by $d_j$.
 
For the structured means we take three degrees of freedom, the smallest integer for which the $t$ distribution has a finite mean and variance while its tails stay heavy enough that a well supported block can escape the shrinkage. Both mean designs have standardized columns, so $b_j=\{n^{-1}\lVert\mathbf Z_j\rVert^2_{\mathbf f}\}^{1/2}$, the root mean square size of the design of expert $j$, is essentially $\sqrt{q_j}$, where $q_j$ is its number of columns and $\lVert\cdot\rVert_{\mathbf f}$ the Frobenius norm. Since $\mathbb E[n^{-1}\lVert\mathbf Z_j\boldsymbol\gamma_j\rVert^2\mid\tau_j]=\tau_j^2b_j^2$, the size of the structured mean is governed by $\tau_jb_j$, not by $\tau_j$ alone, so calibrating $\tau_j$ directly would favor an expert merely for having more columns. We therefore calibrate $\tau_jb_j$ and match its prior median to the expert scale $c_j$ used for the amplitude, placing the mean and the stochastic part of each expert on the same scale. The median of $\tau_j=\lambda_j\lvert t\rvert$, $t\sim t_3$, is $\lambda_jt_{3,0.75}$, so $\lambda_j=c_j/\{t_{3,0.75}\,b_j\}$ makes the prior median of $\tau_jb_j$ equal to $c_j$. The full procedure is summarized in Algorithm~\ref{alg:hyperprior}. The values of $c$, the shares $w_j$, and the correlation targets used in each analysis are reported in the corresponding simulation or application section.
 
\begin{algorithm}[t]
\caption{Design-referenced hyperprior calibration for one expert}
\label{alg:hyperprior}
\begin{algorithmic}[1]
\Require kernel $\mathbf K(l)$, mean design $\mathbf Z$, expert scale $c_j$, correlation targets $u_{\mathrm{lo}}<u_{\mathrm{hi}}$, amplitude factor $r_\sigma$, search bounds $[l_{\min},l_{\max}]$
\State Evaluate $\bar\rho(l)$ on a log-spaced grid over $[l_{\min},l_{\max}]$ and let $[\rho_{\min},\rho_{\max}]$ be its range
\State Solve $\bar\rho(l)=\rho_{\min}+u_k(\rho_{\max}-\rho_{\min})$ for $l_k$, $k\in\{\mathrm{lo},\mathrm{hi}\}$
\State Fit $\log l\sim\mathcal N(\mu_l,s_l^2)$ with $\min(l_{\mathrm{lo}},l_{\mathrm{hi}})$ and $\max(l_{\mathrm{lo}},l_{\mathrm{hi}})$ as its $2.5\%$ and $97.5\%$ quantiles
\State Compute $d\gets n^{-1}\sum_{i=1}^n[\mathbf K(e^{\mu_l})]_{ii}^{1/2}$
\State Fit $\log\sigma\sim\mathcal N(\mu_\sigma,s_\sigma^2)$ with $c_j/(r_\sigma d)$ and $r_\sigma c_j/d$ as its $2.5\%$ and $97.5\%$ quantiles
\State Compute $b\gets\{n^{-1}\lVert\mathbf Z\rVert_{\mathbf f}^2\}^{1/2}$ and set $\lambda\gets c_j/(t_{3,0.75}\,b)$
\Ensure $(\mu_l,s_l)$, $(\mu_\sigma,s_\sigma)$, and $\lambda$
\end{algorithmic}
\end{algorithm}
 
For the intercept, let $(\pi_{\mathrm{lo}},\pi_{\mathrm{hi}})$ be a central $95\%$ prior interval for the baseline risk $\sigma(\beta_0)$, that is, the fraud probability of a unit whose experts all contribute zero. Matching quantiles gives $m_{\beta_0}=\{\operatorname{logit}(\pi_{\mathrm{lo}})+\operatorname{logit}(\pi_{\mathrm{hi}})\}/2$ and $s_{\beta_0}=\{\operatorname{logit}(\pi_{\mathrm{hi}})-\operatorname{logit}(\pi_{\mathrm{lo}})\}/(2z_{0.975})$. We use $(0.01,0.10)$, which gives $m_{\beta_0}=-3.40$, $s_{\beta_0}=0.61$ and a baseline risk of $3.2\%$. For $\eta$ we take $a_\eta,b_\eta\ge1$, which by Remark~\ref{rem:eta-concave} keeps the full conditional log posterior concave. \textcolor{black}{With $b_\eta$ chosen by the percentile rule below, setting $a_\eta=1$ would give a decreasing density with its mode at the origin, favoring the position that no positives have been missed, which is the hypothesis at issue.} We use $a_\eta=2$, the smallest integer above one, and choose $b_\eta$ so that a plausible upper value $\eta_{\mathrm{hi}}$ is the prior $95$th percentile. Taking $\eta_{\mathrm{hi}}=0.5$ encodes that audits are mostly successful and gives $b_\eta=6.39$, with prior mean $0.238$, mode $0.156$, and median $0.217$.

\subsection{Prior sensitivity}
\label{supp:prior-sensitivity}

\textcolor{black}{We assess prior sensitivity over $100$ datasets generated from the design of the illustrative example of Section~\ref{sec:illustration}, fitting the PU--PoE model to each dataset under $13$ prior settings. Starting from the default specification, each of five groups of priors is given a conservative and a diffuse alternative, changed one group at a time, and two joint settings change all five groups simultaneously toward their conservative or diffuse alternatives. Table~\ref{tab:prior-settings} lists the settings. The alternatives change prior location and/or spread, so the comparison is across plausible specifications rather than around a fixed center.}

\begin{table}[tbp]
\centering
\caption{Prior settings of the sensitivity analysis. Here $c$ is the target standard deviation of the total stochastic latent field, the lengthscale pairs are the fractions of the attainable correlation range matched to the prior quantiles, $\lambda_j$ are the Half-$t_3$ scales of Appendix~\ref{supp:prior-elicitation}, and the intercept pairs are central $95\%$ prior intervals for the baseline risk $\sigma(\beta_0)$.}
\label{tab:prior-settings}
\small
\begin{tabular}{llll}
\toprule
Group & Conservative & Default & Diffuse \\
\midrule
GP amplitude, $\sigma_j$ & $c=1$ & $c=2$ & $c=4$ \\
Lengthscale, $l_j$ & $(0.05,0.25)$ & $(0.10,0.50)$ & $(0.10,0.90)$ \\
Structured mean, $\boldsymbol\gamma_j$ & $\lambda_j/2$ & $\lambda_j$ & $2\lambda_j$ \\
Intercept, $\beta_0$ & $(0.005,0.05)$ & $(0.01,0.10)$ & $(0.02,0.20)$ \\
Nondetection rate, $\eta$ & $\mathrm{Beta}(1,5)$ & $\mathrm{Beta}(2,6.39)$ & $\mathrm{Beta}(1,1)$ \\
\bottomrule
\end{tabular}
\end{table}

\textcolor{black}{Figure~\ref{fig:prior-sens-pred} summarizes the posterior probabilities of the latent class. Since the units differ across datasets, they are compared by their common simulation positions rather than individually. The figure averages over the $100$ datasets, at each position, the posterior means and the lower and upper endpoints of the $80\%$ HPD intervals. The shaded bands are therefore averages of interval endpoints, not intervals of a single posterior distribution. The averaged curves are similar across the prior settings, suggesting limited sensitivity of the fitted probabilities. The lengthscale and intercept priors leave the curves essentially unchanged. The amplitude prior has the largest effect, since a larger target scale $c$ pushes the probabilities toward zero and one, and this effect is concentrated on the hidden and observed positives. The structured mean and $\eta$ priors act in the same direction to a smaller extent, and the joint settings combine these shifts.}

\begin{figure}[p]
  \centering
  \includegraphics[width=\linewidth,height=0.84\textheight,keepaspectratio]{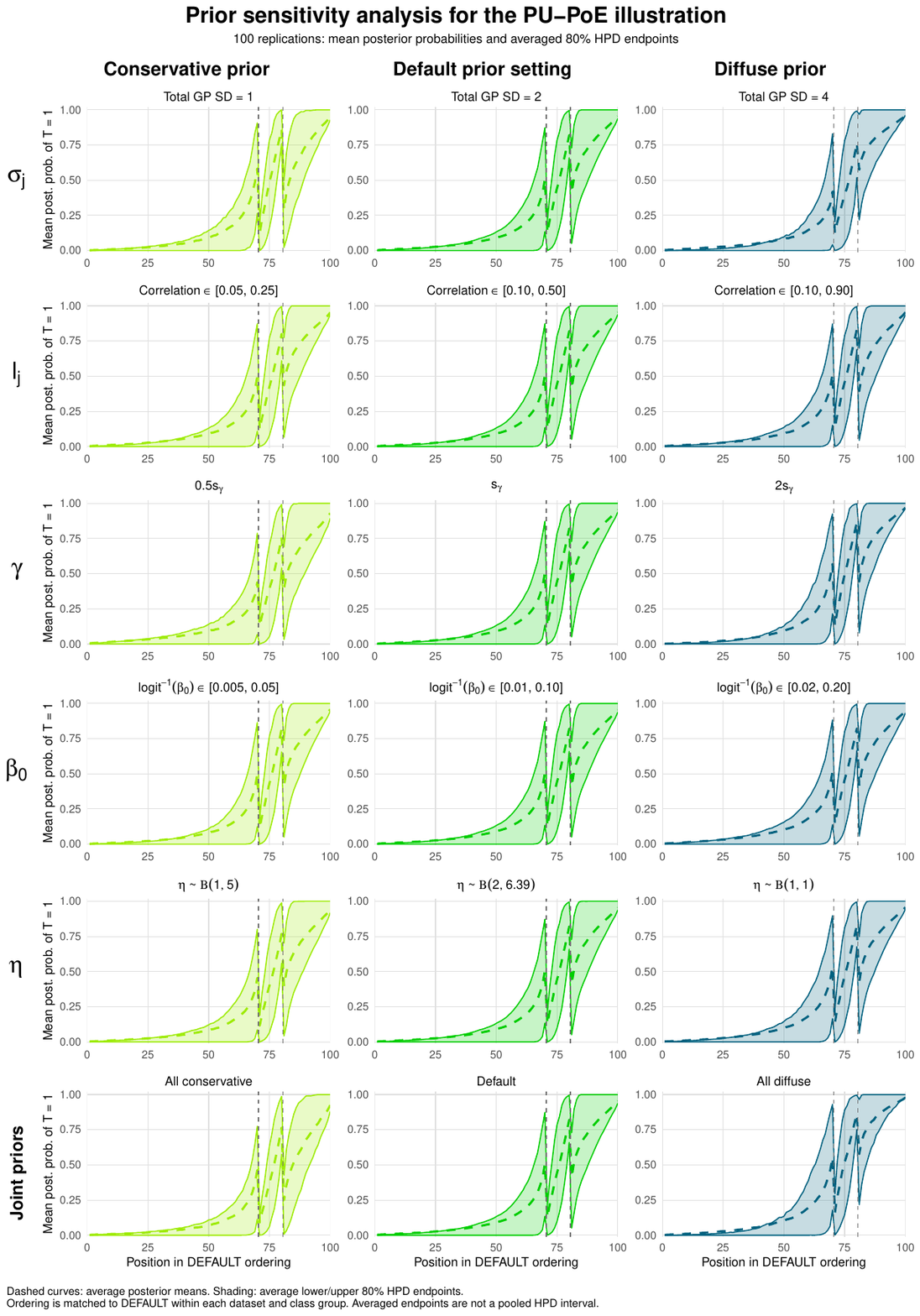}
  \caption{Prior sensitivity of the posterior probabilities over $100$ datasets. Rows correspond to the prior groups of Table~\ref{tab:prior-settings} and columns to the conservative, default and diffuse settings. Dashed curves show the average posterior means at each simulation position and shading the average lower and upper $80\%$ HPD endpoints, which do not form the interval of a single posterior. In each panel, the two vertical lines separate the true zeros (left), the hidden positives (middle) and the observed positives (right).}
  \label{fig:prior-sens-pred}
\end{figure}

\begin{figure}[tbp]
  \centering
  \includegraphics[width=\linewidth]{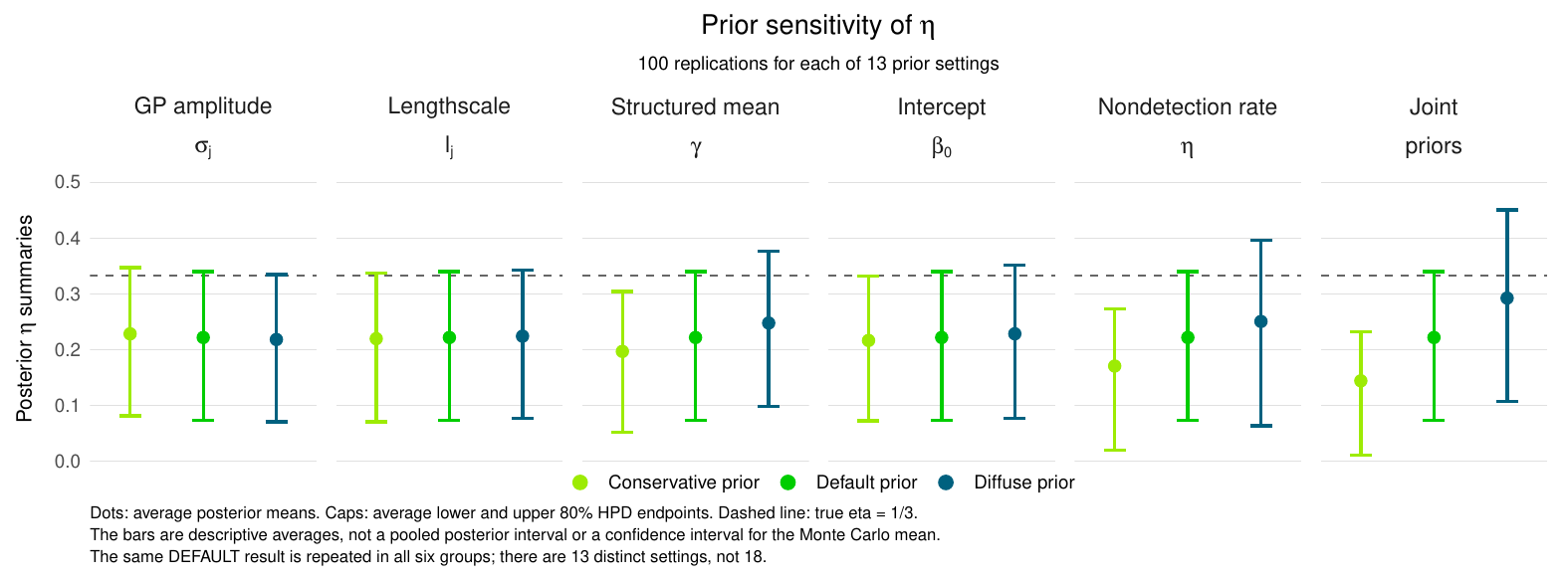}
  \caption{Prior sensitivity of $\eta$ over $100$ datasets. Dots show the average posterior mean and caps the average lower and upper $80\%$ HPD endpoints for each of the $13$ settings, with the default setting repeated in every group. The dashed line marks the true value $\eta=1/3$. The bars are descriptive averages, not the interval of a single posterior or a confidence interval for the Monte Carlo mean.}
  \label{fig:prior-sens-eta}
\end{figure}

\textcolor{black}{Figure~\ref{fig:prior-sens-eta} reports, for each setting, the averages over the $100$ datasets of the posterior mean of $\eta$ and of the endpoints of its $80\%$ HPD interval. Under the default prior, the average posterior mean is $0.22$, with average endpoints $0.07$ and $0.34$, against the true value $\eta=1/3$. The amplitude, lengthscale and intercept priors leave these summaries almost unchanged. The structured mean and the prior on $\eta$ itself have a visible effect, with average posterior means from $0.20$ to $0.25$ when the Half-$t_3$ scales are halved or doubled, and from $0.17$ under $\mathrm{Beta}(1,5)$ to $0.25$ under $\mathrm{Beta}(1,1)$. The joint settings move $\eta$ the most, from $0.15$ under the conservative to $0.29$ under the diffuse specification, whose average endpoints widen to $0.11$ and $0.45$. The underestimation under the default prior agrees with the simulation results of Section~\ref{sec:sim}, and the dependence of $\eta$ on its prior agrees with its weak identification in Section~\ref{sec:eta}, while the ordering of the probabilities is much more stable.}

\textcolor{black}{The prior also affects mixing. Averaged over the $100$ datasets, the mean effective sample size of the fitted probabilities is $2{,}806$ under the default prior. It decreases under the diffuse alternatives of the structured mean, lengthscale and $\eta$ priors, to $2{,}300$, $2{,}635$ and $2{,}524$, respectively, and most under the joint diffuse setting, to $2{,}015$, as weakly identified parameters are allowed to vary more freely. The joint conservative setting also lowers it, to $2{,}509$. The conservative alternatives of the structured mean and $\eta$ priors mix better, at $3{,}297$ and $2{,}924$, but pull $\eta$ further below its true value. The amplitude prior behaves differently, with $1{,}984$ under $c=1$ and $3{,}053$ under $c=4$, while the intercept prior has little effect. These results illustrate a trade-off between mixing and the sensitivity of inference on $\eta$.}

\subsection{Expert decomposition in the illustrative example}
\label{supp:illustration-experts}

\textcolor{black}{Figure~\ref{fig:illustration-experts} decomposes the PU--PoE fit of the illustrative example of Section~\ref{sec:illustration} into its two experts. For each unit it shows the posterior mean and $80\%$ HPD interval of $\sigma(x_{0,i})$ for the covariate expert and of $\sigma(x_{1,i})$ for the network expert, obtained by Gaussian conditioning on $\mathbf f$ as described in Section~\ref{sec:comp}. Both experts assign higher posterior means to the hidden and observed positives than to most true zeros, although the intervals are wide.}

\begin{figure}[tbp]
  \centering
  \includegraphics[width=\linewidth]{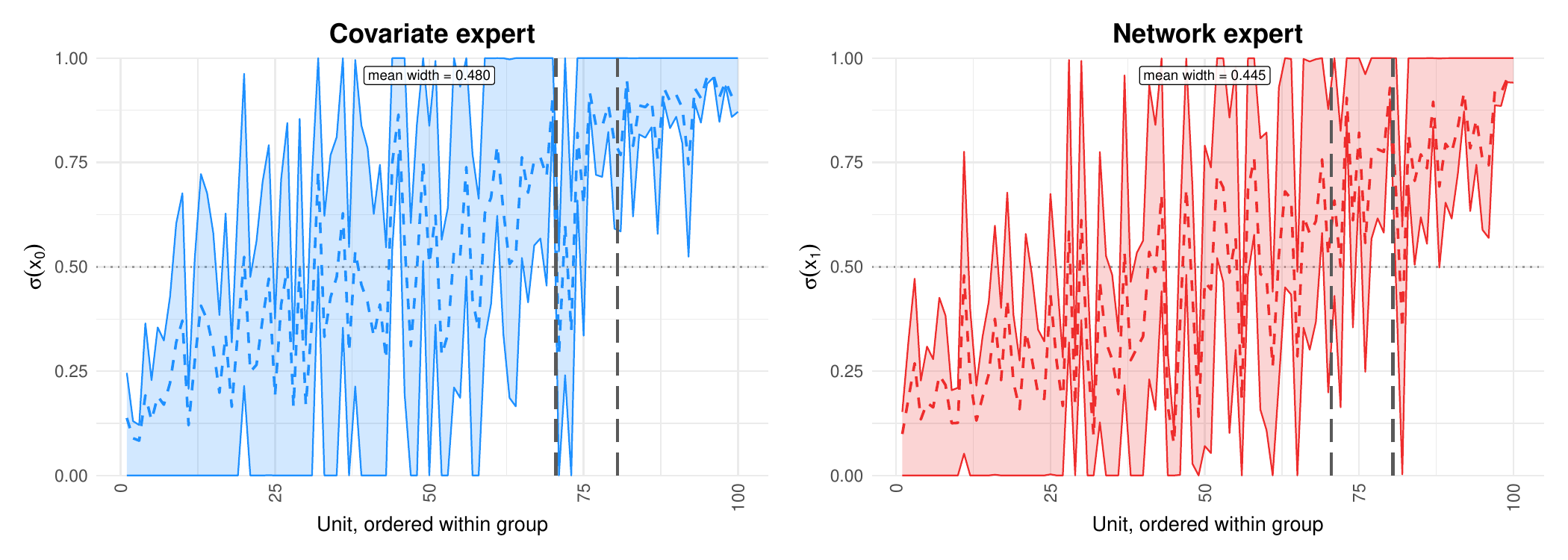}
  \caption{Expert decomposition of the PU--PoE fit in the illustrative example. Posterior means (dashed) and $80\%$ HPD intervals (shaded) of $\sigma(x_{0,i})$ for the covariate expert (left) and $\sigma(x_{1,i})$ for the network expert (right). In each panel, the two vertical lines separate the true negatives (left), the hidden positives (middle) and the observed positives (right), and  units are ordered within each group as in Figure~\ref{fig:illustration} of the main text.}
  \label{fig:illustration-experts}
\end{figure}

\section{Full conditional distributions and implementation details}
\label{supp:mcmc-details}

Each iteration of the sampler updates, in order, the latent field $\mathbf f$ jointly with the kernel hyperparameters, the noise rate, the mean coefficients, and the Half-$t$ scales with their auxiliary variables.
The last two have closed form full conditionals, the noise rate is
updated by slice sampling, and the first block by the auxiliary
gradient move, whose adaptation we detail at the end. We begin with
the closed form updates.

Collect the mean coefficients in
$\mathbf b=(\beta_0,\boldsymbol\gamma_0^\top,\ldots,
\boldsymbol\gamma_m^\top)^\top$ and the corresponding columns in
$\mathbf Z=[\mathbf 1_n,\mathbf Z_0,\ldots,\mathbf Z_m]$, so that
$\boldsymbol\mu_{\mathbf f}=\mathbf Z\mathbf b$. The priors of
Section~\ref{subsec:prior} give
$\mathbf b\sim\mathcal N(\mathbf m_0,\mathbf V_0)$ with
$\mathbf m_0=(m_{\beta_0},\mathbf 0^\top)^\top$ and
$\mathbf V_0=\operatorname{blockdiag}(s_{\beta_0}^2,
\tau_0^2\mathbf I_{q_0},\ldots,\tau_m^2\mathbf I_{q_m})$. Since
$\mathbf f\mid \left( \mathbf b,\boldsymbol\theta \right)
\sim\mathcal N(\mathbf Z\mathbf b,\mathbf C_{\mathbf f})$, with the
labels entering only through $\mathbf f$, conjugacy gives
\[
\mathbf b\mid \left( \mathbf f,\boldsymbol\theta,\boldsymbol\tau \right)
\sim
\mathcal N(\mathbf Q^{-1}\mathbf r_b,\mathbf Q^{-1}),
\qquad
\mathbf Q=\mathbf Z^\top\mathbf C_{\mathbf f}^{-1}\mathbf Z
+\mathbf V_0^{-1},
\qquad
\mathbf r_b=\mathbf Z^\top\mathbf C_{\mathbf f}^{-1}\mathbf f
+\mathbf V_0^{-1}\mathbf m_0.
\]
Reusing the decomposition of $\mathbf C_{\mathbf f}$ from the field
move, the update costs
$\mathcal O(n^2d_b+nd_b^2+d_b^3)$ with $d_b=1+\sum_{j=0}^m q_j$. For $\tau_j^2$ we use the Inverse Gamma representation of the
Half-$t_\nu(0,\lambda_j)$ prior \citep{wand2011},
$\tau_j^2\mid a_j\sim\mathrm{IG}(\nu/2,\nu/a_j)$ with
$a_j\sim\mathrm{IG}(1/2,\lambda_j^{-2})$. Both scales are 
conditionally conjugate, and for a block of $q_j$ coefficients
\[
\begin{aligned}
\tau_j^2\mid\left(\boldsymbol\gamma_j,a_j\right)
&\sim
\mathrm{IG}\!\left(
\frac{\nu+q_j}{2},\
\frac{\nu}{a_j}+\frac{\boldsymbol\gamma_j^\top\boldsymbol\gamma_j}{2}
\right), \\
a_j\mid\tau_j^2
&\sim
\mathrm{IG}\!\left(
\frac{\nu+1}{2},\
\frac{1}{\lambda_j^2}+\frac{\nu}{\tau_j^2}
\right),
\qquad j=0,\ldots,m,
\end{aligned}
\]
so each block is refreshed by two exact draws. Continuing with the noise rate, this is updated by stepping out
slice sampling on $u=\operatorname{logit}(\eta)$, with log target $\mathcal L(\mathbf f;\eta)+\log\pi(\eta)+\log\eta+\log(1-\eta)$ evaluated at $\eta=\{1+\exp(-u)\}^{-1}$, the last two terms being the Jacobian of the transformation. By Remark~\ref{rem:eta-concave},
concavity keeps the slice a single interval.

For the kernel parameters, all proposals are Gaussian random walks on
the log scale, $\boldsymbol\vartheta_j=(\log l_j,\log\sigma_j)^\top$.
We write $\boldsymbol\vartheta$ for the full vector and
$\pi(\boldsymbol\vartheta)$ for its prior density on this scale.
Adaptation uses two warm-up phases, the first to learn a step size
for each expert separately and the second to combine them into one
move over all experts. In the first, the field is updated by the
aGrad-$z$ move \eqref{eq:agrad} and the kernel parameters one expert
at a time, conditionally on the current field. Block $j$ is proposed
as $\boldsymbol\vartheta_j'\sim\mathcal N(\boldsymbol\vartheta_j,
\kappa_j\mathbf D_j\mathbf R_j\mathbf D_j)$. Here $\mathbf D_j$ holds
the prior standard deviations of the block, which set the relative
step of $\sigma_j$ against $l_j$. The correlation matrix
$\mathbf R_j$ is estimated from a rolling window of the path and
mixed with the identity to avoid degeneracy, so proposals travel
along the ridge the two parameters share. The scalar $\kappa_j$ is
adapted on the log scale and carries the overall size. With the field
fixed, the likelihood cancels and only the Gaussian and prior ratios
remain. The step size $\delta$ is adapted in the same phase toward
the acceptance rate recommended by \citet{titsias2018auxiliary} and
then held fixed, so this phase both tunes $\delta$ and learns the
proposal geometry of the kernel parameters.

The second phase moves all kernel parameters at once, through a joint
proposal built from the blocked path. We set
$\boldsymbol\Sigma_{\vartheta,0}
=\operatorname{diag}(\mathbf d_\vartheta)\mathbf R_\vartheta
\operatorname{diag}(\mathbf d_\vartheta)$, where
\[
\mathbf d_\vartheta
=
\left(
\sqrt{\kappa_0}\,s_{l,0},\
\sqrt{\kappa_0}\,s_{\sigma,0},\
\ldots,\
\sqrt{\kappa_m}\,s_{l,m},\
\sqrt{\kappa_m}\,s_{\sigma,m}
\right)^\top
\]
carries the step size learned for each coordinate, with $s_{l,j}$ and
$s_{\sigma,j}$ the prior standard deviations in $\mathbf D_j$, and
$\mathbf R_\vartheta$ is the full correlation matrix estimated from
the tail of the blocked path and treated as above. The proposal is
$\boldsymbol\vartheta'\sim\mathcal N(\boldsymbol\vartheta,
s_\vartheta^2\boldsymbol\Sigma_{\vartheta,0})$, and only the scalar
$s_\vartheta$ is adapted. All proposal parameters are fixed during
production.

\begin{algorithm}[t]
\caption{aGrad-$z$ joint move for $(\mathbf f,\boldsymbol\vartheta)$}
\label{alg:joint}
\begin{algorithmic}[1]
\Require $(\mathbf f,\boldsymbol\vartheta,\eta)$,
$\boldsymbol\mu_{\mathbf f}$, $\delta$, $\boldsymbol\Sigma_\vartheta$,
together with $\mathcal L(\mathbf f;\eta)$,
$\nabla\mathcal L(\mathbf f;\eta)$ and
$\mathbf C_{\mathbf f}(\boldsymbol\vartheta)
=\mathbf U\operatorname{diag}(\xi_i)\mathbf U^\top$ carried from the
previous call
\State Set $\mathbf h=\mathbf f-\boldsymbol\mu_{\mathbf f}$.
\State Draw $\mathbf z\sim\mathcal N\big(
\mathbf h+(\delta/2)\nabla\mathcal L(\mathbf f;\eta),\
(\delta/2)\mathbf I\big)$.
\State Draw $\boldsymbol\vartheta'\sim
\mathcal N(\boldsymbol\vartheta,\boldsymbol\Sigma_\vartheta)$.
\If{$\pi(\boldsymbol\vartheta')=0$}
\hfill{\color{blue}$\triangleright$ outside the prior support}
\State Keep the current state and exit.
\EndIf
\State Form $\mathbf C_{\mathbf f}(\boldsymbol\vartheta')
=\mathbf U'\operatorname{diag}(\xi_i')\mathbf U'^\top$.
\State Set $[\boldsymbol\Lambda_1']_{ii}
=\xi_i'\delta/(\delta+2\xi_i')$.
\State Draw $\boldsymbol\varepsilon\sim
\mathcal N(\mathbf 0,\mathbf I)$.
\State Set $\mathbf h'=\mathbf U'\boldsymbol\Lambda_1'^{1/2}\big(
\boldsymbol\Lambda_1'^{1/2}\mathbf U'^\top(2\mathbf z/\delta)
+\boldsymbol\varepsilon\big)$.
\State Set $\mathbf f'=\boldsymbol\mu_{\mathbf f}+\mathbf h'$.
\State Compute $\mathcal L(\mathbf f';\eta)$ and
$\nabla\mathcal L(\mathbf f';\eta)$.
\State Evaluate $g(\mathbf z,\mathbf h';\mathbf f')$ and
$g(\mathbf z,\mathbf h;\mathbf f)$.
\State Evaluate $Z(\mathbf z,\boldsymbol\vartheta')$ and
$Z(\mathbf z,\boldsymbol\vartheta)$.
\State Set $\log r=\mathcal L(\mathbf f';\eta)
-\mathcal L(\mathbf f;\eta)
+g(\mathbf z,\mathbf h';\mathbf f')
-g(\mathbf z,\mathbf h;\mathbf f)
+\log\dfrac{Z(\mathbf z,\boldsymbol\vartheta')}
{Z(\mathbf z,\boldsymbol\vartheta)}
+\log\dfrac{\pi(\boldsymbol\vartheta')}{\pi(\boldsymbol\vartheta)}$.
\State Draw $v\sim\mathcal U(0,1)$.
\If{$\log v<\min\{0,\log r\}$}
\State Set $(\mathbf f,\boldsymbol\vartheta)\gets
(\mathbf f',\boldsymbol\vartheta')$.
\State Store $\mathbf U'$, $\xi_i'$,
$\mathcal L(\mathbf f';\eta)$,
$\nabla\mathcal L(\mathbf f';\eta)$
\hfill{\color{blue}$\triangleright$ reused by the next call}
\EndIf
\Ensure updated $(\mathbf f,\boldsymbol\vartheta)$
\end{algorithmic}
\end{algorithm}

The joint move draws $\mathbf z$ at the current state as in
\eqref{eq:agrad}, proposes $\boldsymbol\vartheta'$, and draws the
field from the covariance of the proposed kernel. Each proposal
therefore needs a fresh spectral decomposition
$\mathbf C_{\mathbf f}(\boldsymbol\vartheta')
=\mathbf U'\operatorname{diag}(\xi_i')\mathbf U'^\top$, which avoids
inverting $\mathbf C_{\mathbf f}$ since
$\mathbf A(\boldsymbol\vartheta')
=\mathbf U'\boldsymbol\Lambda_1'\mathbf U'^\top$ with
$[\boldsymbol\Lambda_1']_{ii}=\xi_i'\delta/(\delta+2\xi_i')$, and
\[
\mathbf h'\mid\left(\mathbf z,\boldsymbol\vartheta'\right)
\sim
\mathcal N\!\left(
\frac{2}{\delta}\mathbf A(\boldsymbol\vartheta')\mathbf z,\
\mathbf A(\boldsymbol\vartheta')
\right),
\qquad
\mathbf f'=\boldsymbol\mu_{\mathbf f}+\mathbf h'.
\]
Writing
$Z(\mathbf z,\boldsymbol\vartheta)
=\phi(\mathbf z;\mathbf 0,
\mathbf C_{\mathbf f}(\boldsymbol\vartheta)+\delta\mathbf I/2)$
with $\phi$ the Gaussian density, and using the symmetry of the
random walk on $\boldsymbol\vartheta$, the joint log acceptance ratio
is
\[
\log r
=
\mathcal L(\mathbf f';\eta)-\mathcal L(\mathbf f;\eta)
+g(\mathbf z,\mathbf h';\mathbf f')-g(\mathbf z,\mathbf h;\mathbf f)
+\log\frac{Z(\mathbf z,\boldsymbol\vartheta')}
{Z(\mathbf z,\boldsymbol\vartheta)}
+\log\frac{\pi(\boldsymbol\vartheta')}{\pi(\boldsymbol\vartheta)}.
\]
The first two differences are the usual aGrad-$z$ correction for the
field. The $Z$ ratio assesses $\boldsymbol\vartheta'$ against
$\mathbf z$ rather than the field itself, a flatter target that
admits larger moves. Because the $\mathcal O(n^3)$ decomposition is
already paid for, the sampler then refreshes the field at the
accepted kernel at no extra factorization cost.
Algorithm~\ref{alg:joint} gives the joint move step by step and
Algorithm~\ref{alg:schedule} the sampling stage, in which every
proposal parameter is held at its final warm-up value.

\begin{algorithm}[t]
\caption{Sampling stage for the PU--PoE model}
\label{alg:schedule}
\begin{algorithmic}[1]
\Require final warm-up state
$(\mathbf f,\boldsymbol\vartheta,\eta,\mathbf b,\boldsymbol\tau,
\mathbf a)$, the decomposition of
$\mathbf C_{\mathbf f}(\boldsymbol\vartheta)$, $\delta$,
$\boldsymbol\Sigma_\vartheta$, $N$
\For{$t=1,\ldots,N$}
\State Update $(\mathbf f,\boldsymbol\vartheta)$ by
Algorithm~\ref{alg:joint}.
\State Refresh $\mathbf f$ by the aGrad-$z$ move
\eqref{eq:agrad}--\eqref{eq:agrad-acceptance} at the current kernel.
\State Update $\eta$ by stepping out slice sampling on
$\operatorname{logit}(\eta)$.
\State Recompute $\mathcal L(\mathbf f;\eta)$ and
$\nabla\mathcal L(\mathbf f;\eta)$.
\State Draw $\mathbf b\sim
\mathcal N(\mathbf Q^{-1}\mathbf r_b,\mathbf Q^{-1})$.
\State Set $\boldsymbol\mu_{\mathbf f}=\mathbf Z\mathbf b$.
\State Draw $\tau_j^2\mid \left( \boldsymbol\gamma_j,a_j\right)$, $j=0,\ldots,m$.
\State Draw $a_j\mid\tau_j^2$, $j=0,\ldots,m$.
\State Store
$(\mathbf f,\boldsymbol\vartheta,\eta,\mathbf b,\boldsymbol\tau)$.
\EndFor
\Ensure $N$ draws from the posterior \eqref{eq:post-target}
\end{algorithmic}
\end{algorithm}

\section{Simulation study details and additional results}
\label{supp:simulation}

\subsection{The nonlinear design} \label{supp:nonlinear}

Write $\mathcal I_1$ for the $30$ units carrying covariate signal
and $\mathcal I_0$ for the remaining $370$ units. Similarly, let
$\mathcal J_1$ and $\mathcal J_0$ denote the corresponding partition
according to network signal. The ten covariates are generated in five
bivariate blocks, as summarized in Table~\ref{tab:cov-blocks}.
Conditional on the profile indicator shared by the second and fourth
blocks for units in $\mathcal I_1$, the random quantities entering the
five blocks are generated independently. All ten covariate columns are
standardized before model fitting. Only the first block is accessible to a linear classifier. In the
second the two components of the mixture sit symmetrically about the
mean of $\mathcal I_0$, so the two classes share a mean. In the third
they share a center and differ only in radial concentration, in the
fourth they follow the same parabolic trend and differ only in the
displacement, and in the fifth they have identical marginals and
opposite dependence.

\begin{table}[!t]
\centering
\caption{Distributions of the five covariate blocks conditional on
covariate-signal membership.}
\label{tab:cov-blocks}

\footnotesize
\renewcommand{\arraystretch}{1.4}
\setlength{\tabcolsep}{4pt}

\begin{tabular}{
@{}>{\raggedright\arraybackslash}p{0.11\linewidth}
p{0.42\linewidth}
p{0.42\linewidth}@{}
}
\toprule
Block & $\mathcal I_0$ & $\mathcal I_1$ \\
\midrule

Mean shift
&
$\begin{pmatrix}
X_{i1}\\
X_{i2}
\end{pmatrix}
\sim
\mathcal N\!\left(
\begin{pmatrix}
0\\
0
\end{pmatrix},
\begin{pmatrix}
1 & 0\\
0 & 1
\end{pmatrix}
\right)$
&
$\begin{pmatrix}
X_{i1}\\
X_{i2}
\end{pmatrix}
\sim
\mathcal N\!\left(
\begin{pmatrix}
0.5\\
0.5
\end{pmatrix},
\begin{pmatrix}
1 & 0\\
0 & 1
\end{pmatrix}
\right)$
\\
\hdashline

Two profiles
&
$\begin{pmatrix}
X_{i3}\\
X_{i4}
\end{pmatrix}
\sim
\mathcal N\!\left(
\begin{pmatrix}
0\\
1
\end{pmatrix},
\begin{pmatrix}
2^2 & 0\\
0 & 1.2^2
\end{pmatrix}
\right)$
&
$\begin{aligned}
\begin{pmatrix}
X_{i3}\\
X_{i4}
\end{pmatrix}
&\sim
\mathcal N\!\Biggl(
\begin{pmatrix}
0\\
1
\end{pmatrix}
+
\xi_i
\begin{pmatrix}
2.6\\
-1.4
\end{pmatrix}, \\
&\quad
\begin{pmatrix}
0.9^2 & 0\\
0 & 0.7^2
\end{pmatrix}
\Biggr)
\end{aligned}$
\\
\hdashline

Ring
&
$\begin{aligned}
(X_{i5},X_{i6})
&=R_i(\cos\vartheta_i,\sin\vartheta_i),\\
\vartheta_i
&\sim \operatorname{Unif}(0,2\pi),\\
U_i
&\sim \operatorname{Unif}(0,1),\\
R_i
&=2\sqrt{U_i}
\end{aligned}$
&
$\begin{aligned}
(X_{i5},X_{i6})
&=R_i(\cos\vartheta_i,\sin\vartheta_i),\\
\vartheta_i
&\sim \operatorname{Unif}(0,2\pi),\\
R_i
&\sim \operatorname{TN}_{[0,2]}(1.5,0.25^2)
\end{aligned}$
\\
\hdashline

Parabola
&
$\begin{aligned}
X_{i7}
&\sim \mathcal N(0,1),\\
\varepsilon_i
&\sim \mathcal N(0,0.4^2),\\
X_{i8}
&=0.35X_{i7}^2+\varepsilon_i
\end{aligned}$
&
$\begin{aligned}
X_{i7}
&\sim \mathcal N(0,1),\\
\varepsilon_i
&\sim \mathcal N(0,0.35^2),\\
X_{i8}
&=0.35X_{i7}^2-0.8\,\xi_i+\varepsilon_i
\end{aligned}$
\\
\hdashline

Opposite correlation
&
$\begin{pmatrix}
X_{i9}\\
X_{i10}
\end{pmatrix}
\sim
\mathcal N\!\left(
\begin{pmatrix}
0\\
0
\end{pmatrix},
\begin{pmatrix}
1 & 0.85\\
0.85 & 1
\end{pmatrix}
\right)$
&
$\begin{pmatrix}
X_{i9}\\
X_{i10}
\end{pmatrix}
\sim
\mathcal N\!\left(
\begin{pmatrix}
0\\
0
\end{pmatrix},
\begin{pmatrix}
1 & -0.85\\
-0.85 & 1
\end{pmatrix}
\right)$
\\

\bottomrule
\end{tabular}

\vspace{0.7ex}

\begin{minipage}{0.94\linewidth}
\footnotesize
\textit{Note:}
$\mathcal I_1$ contains the $30$ covariate-signal units and
$\mathcal I_0$ the remaining $370$.
For $i\in\mathcal I_1$,
$\Pr(\xi_i=-1)=\Pr(\xi_i=1)=1/2$, with the same $\xi_i$
used in the two-profiles and parabola blocks.
$\operatorname{TN}_{[a,b]}(\mu,\sigma^2)$ denotes a
$\mathcal N(\mu,\sigma^2)$ distribution truncated to $[a,b]$.
\end{minipage}

\end{table}

The network is built in three layers and then merged. Each layer
starts from a CRP with concentration $\alpha=2$, the first over all
$400$ units, the second over $\mathcal J_1$ and the third over
$\mathcal J_0$. Inside a group of size $g$ produced by the CRP every
pair is joined independently with probability
$\min\{1,\,2/(g-1)\}$, so each unit gains about two connections
within its group whether the group is small or large. An edge is
kept if it appears in at least one layer, and the network is then
closed once over triangles, so that units sharing neighbors are
given a chance to become adjacent themselves. Finally, each shared neighbor gives an independent closure probability of
$0.05$, so a pair with $c$ neighbors in common is joined with
probability $1-0.95^{c}$, which grows with the number of neighbors
they share.

Table~\ref{tab:net-summaries} summarizes the result over $1{,}000$
generated networks. None of the three partitions collapses into a
single group, and the second and third layers add the same number of
connections per unit, so both classes end up with similar degrees.
The class signal is therefore not in how many connections a unit has
but in whom it connects to, with a unit in $\mathcal J_1$ having ten
times as many neighbors in $\mathcal J_1$ as a unit outside it.

\begin{table}[tbp]
\centering
\caption{Network summaries over $1{,}000$ generated networks, means
with standard deviations in parentheses.}
\label{tab:net-summaries}

\footnotesize
\renewcommand{\arraystretch}{1.25}
\begin{tabular}{lcc}
\toprule
& $\mathcal J_1$ & $\mathcal J_0$ \\
\midrule
CRP groups in the shared layer
& \multicolumn{2}{c}{$11.09$ $(2.91)$} \\
CRP groups in the own layer
& $6.09$ $(2.02)$ & $10.99$ $(3.03)$ \\
\midrule
Edge density within class
& $0.0717$ $(0.0128)$ & $0.0123$ $(0.0005)$ \\
Edge density across classes
& \multicolumn{2}{c}{$0.0063$ $(0.0009)$} \\
\midrule
Mean degree
& $4.42$ $(0.53)$ & $4.72$ $(0.21)$ \\
neighbors belonging to $\mathcal J_1$
& $2.08$ $(0.37)$ & $0.19$ $(0.03)$ \\
\bottomrule
\end{tabular}
\end{table}

Figure~\ref{fig:sim-dgp} shows the class conditional covariate densities on the left and network summaries based on $1{,}000$ generated networks on the right, including edge probabilities by class pair, the degree distributions of the two classes, the share of neighbors that are positives, and a representative network shown both as a graph and as an adjacency matrix.

\begin{figure}[tbp]
  \centering
  \includegraphics[width=0.775\linewidth]
    {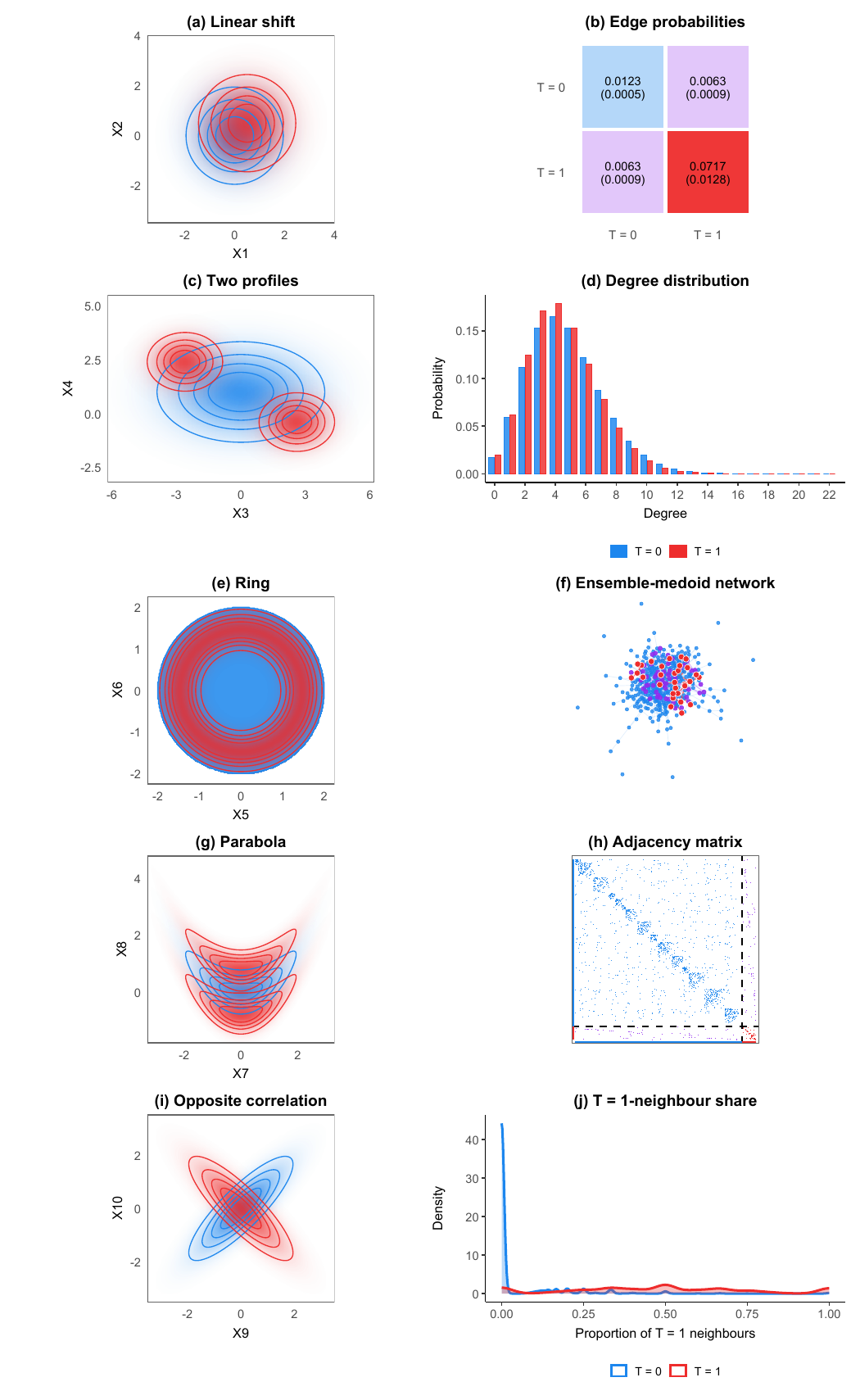}
    
  \caption{Nonlinear design. Class conditional covariate densities are shown on the left. Network summaries on the right are based on $1{,}000$ generated networks, except panels (f) and (h), which show the ensemble medoid network and its adjacency matrix.}
  \label{fig:sim-dgp}
\end{figure}

\subsection{The linear design} \label{supp:linear}

The linear design keeps the sample size, the class allocation and the
hiding rate of the nonlinear one, and replaces both generating
mechanisms by linear ones. The ten covariates are drawn
independently, $\mathcal N(0.5,1)$ for the $30$ units with covariate
signal and $\mathcal N(0,1)$ for the remaining $370$, so a single
hyperplane separates the classes. The network is a two block SBM in
which the $30$ units carrying network signal form one block, with
edge probabilities $0.25$ within that block, $0.10$ within the
remaining units and $0.08$ across, so the structure is assortative. The linear design is considered only under the SCAR hiding mechanism, as described in Subsection~\ref{supp:hiding}.

\subsection{The two hiding mechanisms} \label{supp:hiding}

Under random hiding, five of the fifteen positives are drawn
uniformly without replacement from each of the three signal groups,
so the hidden set is independent of the realized covariates and
network. This is the mechanism used in both designs.

Covariate dependent hiding applies to the nonlinear design and
selects the hidden positives through an oracle score built from the
true generating densities of Section~\ref{supp:nonlinear}. For block
$b$ let $f_b(\cdot\mid T=1)$ and $f_b(\cdot\mid T=0)$ be the class
conditional densities, where for the second and fourth blocks the
positive density is the equal mixture over the two profiles. The
score of unit $i$ is the total log likelihood ratio,
\[
\psi_i=\sum_{b=1}^{5}\log\frac{f_b(\mathbf x_{ib}\mid T_i=1)}
{f_b(\mathbf x_{ib}\mid T_i=0)},
\]
so a large $\psi_i$ means that the covariates of unit $i$ look
strongly fraudulent under the true model. The score is standardized
across the $45$ positives and clipped,
\[
\tilde\psi_i=\max\big[-3,\min\{3,
(\psi_i-\bar\psi)/\mathrm{sd}(\psi)\}\big],
\]
and the fifteen hidden positives are drawn without replacement with
probability proportional to $\exp(-\lambda\tilde\psi_i)$ with
$\lambda=2$. The dependence on the covariates is what makes the
mechanism SAR rather than SCAR. The exponent is negative because this
weight is the propensity to be hidden, so the positives with the
weakest covariate signal are the most likely to go unrecorded and the
ones left recorded are the easiest to find. Setting $\lambda=0$
returns the random mechanism, and the score is used only to generate
the labels.

\subsection{Competing method settings} \label{supp:competitors}

AdaSampling \citep{yang2018adasampling} is fitted with the \textsf{AdaSampling} package \citep{pkg_AdaSampling}, using $C=50$ SVM classifiers and a sample factor of one. Weighted logistic regression is an unpenalized binomial GLM with logit link, unit weight on the observed positives and weight $1/2$ on the zero labeled units. INLA \citep{rue2009approximate} is fitted with the \textsf{INLA} package \citep{pkg_INLA} as a binomial logistic regression in which the intercept enters as an explicit constant column, so that all coefficients, the intercept included, receive the default Gaussian prior with mean zero and precision $0.001$.

The GCN \citep{KipfWelling2017GCN} has two graph convolutional layers with $64$ hidden ReLU units and dropout $0.5$, uses the symmetrically normalized adjacency matrix with self loops, and is trained on the observed labels for $200$ epochs with Adam, learning rate $10^{-2}$ and weight decay $5\times10^{-4}$. The nnPU classifier \citep{kiryo2017nnpu} has one hidden layer of $100$ ReLU units and uses the sigmoid surrogate loss, trained with $5{,}000$ full batch Adam updates at learning rate $10^{-3}$, weight decay $0.005$ and correction parameters $\beta=0$ and $\gamma=1$. Its positive sample consists of the $30$ observed positives and its unlabeled sample of all $400$ units, and its class prior is set to the true value $\pi_p=45/400$. Ranking metrics use its raw decision scores, while log loss and Brier scores use their logistic transformation. For AdaSampling, nnPU, WLR and INLA, network information is added through the standardized modularity eigenvectors.

The linear Bayesian models keep only the intercept and the structured mean terms of the corresponding GP models using the same priors. Namely, $\beta_0\sim\mathcal N(-3.40,0.61^2)$, the Half-$t_3$ block scales of Appendix~\ref{supp:prior-elicitation} recalibrated for each replication, and $\eta\sim\mathrm{Beta}(2,6.39)$ in the PU versions, while the posterior sampling is carried out by P\'olya--Gamma data augmentation \citep{polson2013bayesian}. In their network versions, the structured mean uses the same standardized modularity eigenvectors. Finally, WLR and INLA produce predictive label draws in the same way as the Bayesian models, using draws of the linear predictor from its asymptotic Gaussian distribution for WLR and from its posterior marginals for INLA.

\subsection{Additional results} \label{supp:sim}

Table~\ref{tab:sim-linear-scar} reports all metrics for the linear design under SCAR. Coverage and interval scores are unavailable for Ada, nnPU and the GCN, which do not produce predictive draws for the latent labels. For the Bayesian models, predictive label draws are generated as $T_i^{(s)}\sim\mathrm{Bernoulli}\{\sigma(f_i^{(s)})\}$, so the coverage and interval scores target the same latent class $T_i$ in the PU and non-PU specifications and can be compared directly.

\begin{table}[!t]
\centering
\caption{Linear design under SCAR. Monte Carlo means over $100$
replications, standard deviations in parentheses. Best value in bold,
closest to nominal for coverage.}
\label{tab:sim-linear-scar}
\begingroup
\scriptsize
\setlength{\tabcolsep}{1.4pt}
\renewcommand{\arraystretch}{1.10}
\setlength{\aboverulesep}{0.15ex}
\setlength{\belowrulesep}{0.15ex}

\resizebox{\linewidth}{!}{
\begin{tabular}{lcccccc}
\toprule
Method
& \shortstack{Overall\\AUC}
& \shortstack{Hidden\\AUC}
& \shortstack{Overall\\Recall@45}
& \shortstack{Hidden\\Recall@15}
& \shortstack{Overall\\LogLoss}
& \shortstack{Hidden\\LogLoss} \\
\midrule
Lin-Cov & 0.766 (0.031) & 0.713 (0.062) & 0.406 (0.055) & 0.219 (0.089) & 0.310 (0.013) & 2.418 (0.177) \\
Lin-PU-Cov & 0.766 (0.031) & 0.712 (0.062) & 0.404 (0.053) & 0.218 (0.084) & 0.301 (0.013) & 2.202 (0.186) \\
Lin-Cov+Net & 0.806 (0.041) & 0.754 (0.070) & 0.451 (0.071) & 0.255 (0.114) & 0.291 (0.021) & 2.301 (0.237) \\
Lin-PU-Cov+Net & 0.808 (0.042) & 0.754 (0.070) & 0.450 (0.070) & 0.251 (0.116) & 0.281 (0.023) & 2.052 (0.265) \\
\hdashline
GP-Cov & 0.846 (0.028) & 0.705 (0.064) & 0.540 (0.051) & 0.202 (0.084) & 0.273 (0.013) & 2.513 (0.195) \\
GP-PU-Cov & 0.848 (0.027) & 0.704 (0.063) & 0.533 (0.052) & 0.198 (0.087) & 0.263 (0.013) & 2.205 (0.212) \\
GP-PoE & 0.941 (0.018) & 0.826 (0.053) & 0.774 (0.043) & \textbf{0.337} (0.125) & 0.195 (0.011) & 2.426 (0.204) \\
GP-PU-PoE & 0.939 (0.018) & 0.821 (0.054) & 0.764 (0.043) & 0.333 (0.120) & \textbf{0.191} (0.012) & 2.087 (0.237) \\
GP-PU-PoE$_{\beta_0}$ & \textbf{0.943} (0.017) & \textbf{0.830} (0.052) & \textbf{0.778} (0.037) & \textbf{0.337} (0.110) & 0.200 (0.009) & 2.233 (0.139) \\
\hdashline
Ada & 0.786 (0.118) & 0.689 (0.093) & 0.471 (0.099) & 0.178 (0.089) & 0.479 (0.038) & 0.884 (0.130) \\
Ada+Net & 0.837 (0.075) & 0.742 (0.079) & 0.534 (0.084) & 0.210 (0.108) & 0.456 (0.040) & \textbf{0.813} (0.165) \\
nnPU & 0.729 (0.040) & 0.688 (0.069) & 0.411 (0.058) & 0.169 (0.080) & 0.431 (0.047) & 4.370 (0.494) \\
nnPU+Net & 0.774 (0.053) & 0.718 (0.075) & 0.485 (0.092) & 0.186 (0.100) & 0.363 (0.066) & 4.078 (0.610) \\
WLR & 0.766 (0.031) & 0.708 (0.063) & 0.400 (0.051) & 0.207 (0.088) & 0.299 (0.013) & 1.905 (0.239) \\
WLR+Net & 0.808 (0.042) & 0.750 (0.069) & 0.446 (0.067) & 0.242 (0.111) & 0.278 (0.024) & 1.838 (0.288) \\
INLA & 0.765 (0.031) & 0.709 (0.063) & 0.405 (0.052) & 0.209 (0.088) & 0.315 (0.013) & 2.623 (0.276) \\
INLA+Net & 0.806 (0.042) & 0.749 (0.069) & 0.450 (0.068) & 0.241 (0.114) & 0.297 (0.022) & 2.585 (0.351) \\
GCN & 0.799 (0.034) & 0.746 (0.074) & 0.408 (0.052) & 0.191 (0.093) & 0.313 (0.008) & 2.324 (0.145) \\
\bottomrule
\end{tabular}
}

\vspace{0.3ex}

\resizebox{\linewidth}{!}{
\begin{tabular}{lcccccc}
\toprule
Method
& \shortstack{Overall\\Cov$_{80}$}
& \shortstack{Hidden\\Cov$_{80}$}
& \shortstack{Overall\\IS$_{80}$}
& \shortstack{Hidden\\IS$_{80}$}
& \shortstack{Overall\\Brier}
& \shortstack{Hidden\\Brier} \\
\midrule
Lin-Cov & \textbf{0.951} (0.007) & 0.488 (0.116) & 0.704 (0.066) & 5.608 (1.041) & 0.089 (0.003) & 0.778 (0.033) \\
Lin-PU-Cov & 0.962 (0.006) & 0.575 (0.112) & 0.685 (0.058) & 4.822 (1.008) & 0.087 (0.004) & 0.728 (0.042) \\
Lin-Cov+Net & 0.957 (0.008) & 0.527 (0.123) & 0.650 (0.078) & 5.254 (1.109) & 0.085 (0.006) & 0.747 (0.053) \\
Lin-PU-Cov+Net & 0.968 (0.008) & 0.628 (0.124) & 0.626 (0.076) & 4.348 (1.116) & 0.081 (0.007) & 0.682 (0.071) \\
\hdashline
GP-Cov & 0.966 (0.007) & 0.443 (0.129) & 0.553 (0.063) & 6.010 (1.158) & 0.080 (0.004) & 0.796 (0.032) \\
GP-PU-Cov & 0.978 (0.006) & 0.582 (0.120) & 0.552 (0.056) & 4.762 (1.082) & 0.076 (0.004) & 0.729 (0.046) \\
GP-PoE & 0.980 (0.005) & 0.479 (0.133) & \textbf{0.350} (0.045) & 5.692 (1.195) & 0.057 (0.004) & 0.791 (0.041) \\
GP-PU-PoE & 0.986 (0.005) & 0.619 (0.126) & 0.386 (0.045) & 4.432 (1.132) & \textbf{0.053} (0.003) & 0.712 (0.063) \\
GP-PU-PoE$_{\beta_0}$ & 0.984 (0.004) & 0.575 (0.113) & 0.353 (0.047) & 4.828 (1.019) & 0.056 (0.003) & 0.772 (0.028) \\
\hdashline
Ada & -- & -- & -- & -- & 0.152 (0.015) & 0.325 (0.049) \\
Ada+Net & -- & -- & -- & -- & 0.144 (0.016) & \textbf{0.293} (0.063) \\
nnPU & -- & -- & -- & -- & 0.083 (0.010) & 0.892 (0.049) \\
nnPU+Net & -- & -- & -- & -- & 0.070 (0.011) & 0.866 (0.059) \\
WLR & 0.979 (0.004) & 0.739 (0.099) & 0.714 (0.060) & 3.345 (0.887) & 0.085 (0.004) & 0.629 (0.054) \\
WLR+Net & 0.982 (0.005) & \textbf{0.763} (0.102) & 0.642 (0.077) & \textbf{3.133} (0.915) & 0.080 (0.007) & 0.599 (0.072) \\
INLA & 0.953 (0.006) & 0.495 (0.110) & 0.698 (0.057) & 5.548 (0.988) & 0.087 (0.003) & 0.769 (0.043) \\
INLA+Net & 0.957 (0.007) & 0.518 (0.117) & 0.648 (0.082) & 5.337 (1.051) & 0.083 (0.006) & 0.743 (0.063) \\
GCN & -- & -- & -- & -- & 0.093 (0.002) & 0.796 (0.024) \\
\bottomrule
\end{tabular}
}
\endgroup

\end{table}

\begin{table}[!t]
\centering
\caption{Estimation of $\eta$ and $H$ in the linear design under SCAR over $100$ replications.}
\label{tab:eta-linear-scar}
\small
\setlength{\tabcolsep}{6pt}
\begin{tabular}{lrrrrrr}
\toprule
Method
& Mean $\hat{\eta}$ & RMSE$_{\eta}$ & Cov$_{80,\eta}$ & IS$_{80,\eta}$ & Mean $\hat H$ & RMSE$_H$ \\
\midrule
GP-PU-PoE
& 0.209 & 0.127 & 0.870 & 0.329 & 7.969 & 7.149 \\
GP-PU-PoE$_{\beta_0}$
& 0.181 & 0.153 & 0.450 & 0.386 & 6.637 & 8.390 \\
\bottomrule
\end{tabular}

\end{table}

The proposed model comes first on the linear design as well, although this design is favorable to the competing methods, since the covariates are linearly separable and the network is a stochastic block model. It reaches an overall AUC of $0.943$ and a hidden AUC of $0.830$, against $0.808$ and $0.754$ for the best linear Bayesian model and $0.837$ and $0.750$ for the best competitor. Both sources contribute, the covariates raising the overall AUC from $0.766$ to $0.848$ once read through a Gaussian process, and the network expert a further $0.091$ on top of that, to $0.939$. Regarding $\eta$, the contribution of the structured mean is visible here, as Table~\ref{tab:eta-linear-scar} shows. The full model gives a mean $\hat\eta$ of $0.209$ with coverage $0.870$, close to the nominal $0.80$, while GP-PU-PoE$_{\beta_0}$, which replaces the structured mean by an intercept, gives $0.181$ with coverage $0.450$.

\section{Additional application results}
\label{supp:application}

\subsection{Elicited priors} \label{supp:elicited-priors}

\textcolor{black}{For the application, we use the index $j=0$ for the covariate expert and $j=1,2,3$ for Networks~1--3, respectively.} Applying Algorithm~\ref{alg:hyperprior} to the observed designs with $c=2$ and shares $(1/2,1/6,1/6,1/6)$ gives the expert scales $c_0=\sqrt{2}$ for the covariates and $c_j=\sqrt{6}/3$ for each network. The resulting expert specific hyperpriors are reported in Table~\ref{tab:supp-priors}. The attainable correlation ranges support the use of relative rather than common absolute targets. The covariate kernel spans the whole unit interval, so the $10\%$ and $50\%$ levels correspond to absolute targets of $0.100$ and $0.500$. For the network kernels, the attainable maxima are $0.801$, $0.324$ and $0.397$ for \textcolor{black}{Network~1, Network~2 and Network~3}, respectively. Since the attainable minima are essentially zero, the same relative levels give the target pairs $(0.080,0.400)$,
$(0.032,0.162)$ and $(0.040,0.198)$. For the four mean designs, whose columns are standardized,
$b_j=\{n^{-1}\lVert\mathbf Z_j\rVert_{\mathrm F}^{2}\}^{1/2}$ gives $3.602$, $2.234$, $1.413$ and $0.999$. We set $\lambda_j=c_j/(t_{3,0.75}b_j)$ so that the prior median of $\tau_jb_j$ equals the one standard deviation expert scale $c_j$. Using the rounded values of $\lambda_j$ gives medians $1.405$, $0.820$, $0.821$ and $0.818$, matching the targets $1.414$ for the
covariate expert and $0.817$ for each network expert. Finally, the priors for the intercept and the PU parameter are $\beta_0\sim\mathcal N(-3.4,0.61^2)$ and
$\eta\sim\mathrm{Beta}(2,6.39)$, respectively.
\begin{table}[tbp]
\centering
\caption{Elicited hyperpriors for the four experts.}
\label{tab:supp-priors}
\footnotesize
\setlength{\tabcolsep}{4pt}
\resizebox{\linewidth}{!}{
\begin{tabular}{lrrrrrrrr}
\toprule
Expert & $q_j$ & Corr.\ targets & $l$ median & $l$ interval &
$\sigma$ median & $\sigma$ interval & $b_j$ & $\lambda_j$ \\
\midrule
Covariates & 13 & 0.100, 0.500 & 2.746 & [1.892, 3.984] &
1.419 & [0.715, 2.818] & 3.602 & 0.51 \\
Network 1 & 5 & 0.080, 0.400 & 0.317 & [0.282, 0.356] &
6.297 & [3.171, 12.503] & 2.234 & 0.48 \\
Network 2 & 2 & 0.032, 0.162 & 0.350 & [0.305, 0.401] &
4.712 & [2.373, 9.356] & 1.413 & 0.76 \\
Network 3 & 1 & 0.040, 0.198 & 0.343 & [0.305, 0.386] &
5.366 & [2.702, 10.655] & 0.999 & 1.07 \\
\bottomrule
\end{tabular}
}

\end{table}

\subsection{Posterior summaries and diagnostics}
\label{supp:posterior-diagnostics}

Figure~\ref{fig:supp-traces-full} shows the parameter traces for the
full PU-PoE model. Table~\ref{tab:supp-parameters} reports the
posterior summaries and effective sample sizes of the sampled
parameters, while Table~\ref{tab:supp-ess} summarizes the effective
sample sizes for the latent field and the induced probabilities.
Figure~\ref{fig:supp-eta} shows a histogram of the posterior of
$\eta$.

Furthermore, as a diagnostic for the anchor condition in
Assumption~\ref{ass:lower-anchor}, at each posterior draw we compared
the maximum fitted latent probability, $\max_i \sigma(f_i)$, with the
threshold $0.99$. The maximum exceeded $0.99$ in every retained draw,
giving a posterior probability of $1.000$, while its posterior
$2.5\%$ quantile was also $1.000$ to three decimal places. Thus, every posterior draw contains at least one unit with near certain latent risk. \textcolor{black}{This finding is compatible with the anchor condition in Assumption~\ref{ass:lower-anchor} under the fitted model and priors, but does not independently verify the condition.}

\begin{table}[tbp]
\centering
\caption{Posterior summaries and effective sample sizes for the full
PU-PoE model.}
\label{tab:supp-parameters}
\small
\setlength{\tabcolsep}{4pt}
\begin{tabular}{lrrrrrr}
\toprule
Parameter & Mean & SD & 5\% & Median & 95\% & ESS \\
\midrule
$l_{\mathrm{0}}$ & 2.884 & 0.488 & 2.172 & 2.836 & 3.744 & 1399.5 \\
$\sigma_{\mathrm{0}}$ & 1.592 & 0.539 & 0.865 & 1.505 & 2.558 & 601.1 \\
$l_{\mathrm{1}}$ & 0.312 & 0.018 & 0.283 & 0.311 & 0.343 & 1359.5 \\
$\sigma_{\mathrm{1}}$ & 7.894 & 2.754 & 4.236 & 7.461 & 12.912 & 760.5 \\
$l_{\mathrm{2}}$ & 0.342 & 0.023 & 0.306 & 0.341 & 0.382 & 1242.7 \\
$\sigma_{\mathrm{2}}$ & 6.752 & 2.209 & 3.681 & 6.466 & 10.684 & 903.5 \\
$l_{\mathrm{3}}$ & 0.336 & 0.020 & 0.304 & 0.336 & 0.370 & 915.8 \\
$\sigma_{\mathrm{3}}$ & 8.247 & 4.014 & 3.677 & 7.300 & 16.341 & 318.5 \\
$\eta$ & 0.137 & 0.076 & 0.029 & 0.129 & 0.277 & 1301.9 \\
$\beta_0$ & $-3.586$ & 0.544 & $-4.476$ & $-3.598$ & $-2.681$ & 2694.6 \\
$\tau_{\mathrm{0}}$ & 0.268 & 0.193 & 0.022 & 0.236 & 0.630 & 996.1 \\
$\tau_{\mathrm{1}}$ & 0.398 & 0.387 & 0.034 & 0.311 & 1.043 & 257.1 \\
$\tau_{\mathrm{2}}$ & 0.793 & 0.780 & 0.055 & 0.582 & 2.184 & 2089.4 \\
$\tau_{\mathrm{3}}$ & 1.035 & 1.128 & 0.062 & 0.740 & 2.960 & 1901.1 \\
\bottomrule
\end{tabular}

\end{table}

\begin{table}[tbp]
\centering
\caption{Effective sample size summaries for the latent field and the
induced probabilities.}
\label{tab:supp-ess}
\small
\begin{tabular}{lrrrrrr}
\toprule
Quantity & Mean ESS & Median ESS & Min ESS & 5\% ESS & 95\% ESS & Max ESS \\
\midrule
$f$ & 1406.7 & 1324.6 & 152.7 & 787.6 & 2233.2 & 3477.0 \\
$\sigma(f)$ & 3947.4 & 4089.3 & 771.8 & 2659.1 & 4742.3 & 5000.0 \\
\bottomrule
\end{tabular}

\end{table}

\begin{figure}[tbp]
\centering
\includegraphics[width=0.90\linewidth]
  {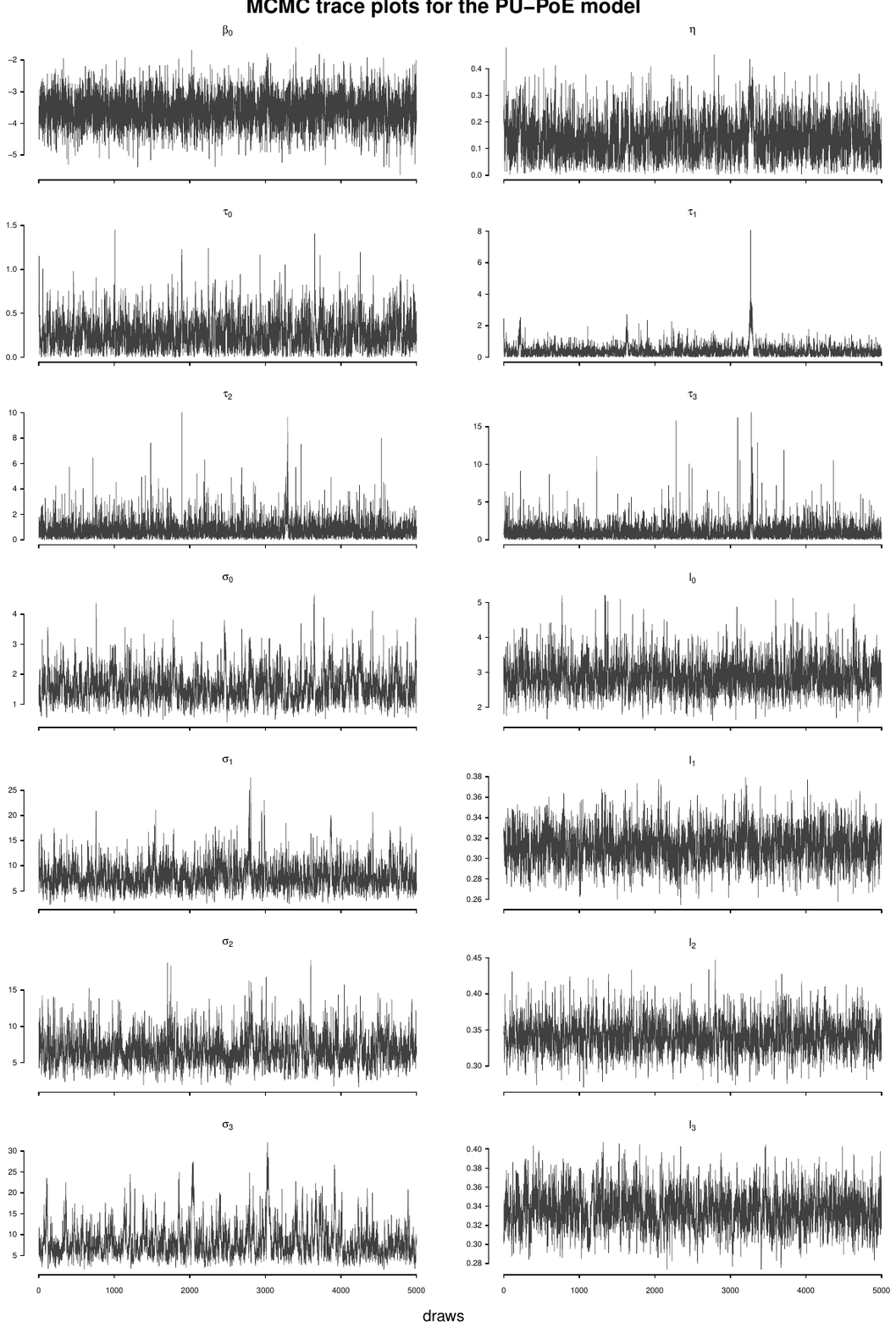}

\caption{Parameter trace plots for the full PU-PoE model.}
\label{fig:supp-traces-full}
\end{figure}

\begin{figure}[!t]
\centering
\includegraphics[width=0.6\linewidth]
  {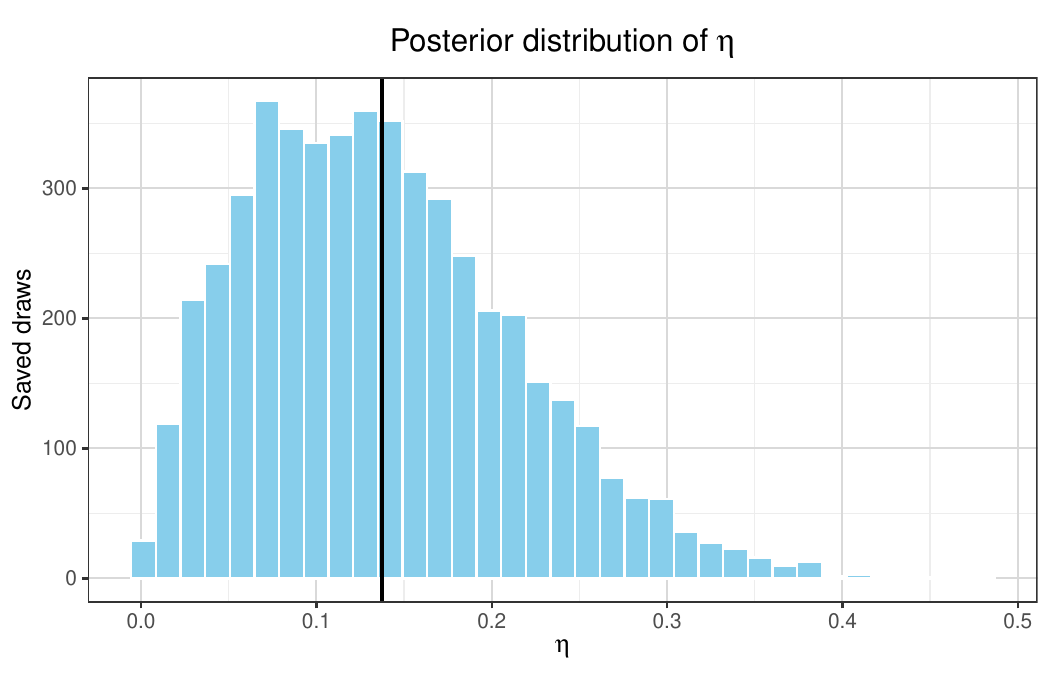}
  
\caption{Posterior distribution of $\eta$ under the full PU-PoE
model. The vertical line marks the posterior mean.}
\label{fig:supp-eta}
\end{figure}

\subsection{Effective expert shares}
\label{supp:expert-shares}

Table~\ref{tab:expert-weights} reports the prior allocation and the
posterior effective share of each expert. For posterior draw $s$, we
measure the GP component of expert $j$ by $v_{K,j}^{(s)}
=
n^{-1}\operatorname{tr}\{\mathbf C_j^{(s)}\}$
the average marginal variance of its Gaussian process, and its
structured mean component by
$v_{\mu,j}^{(s)}
=
\operatorname{Var}_i
\{[\mathbf Z_j\boldsymbol\gamma_j^{(s)}]_i\}$, the variance of its mean contribution across \textcolor{black}{firms}. The covariance,
mean and combined shares are respectively
$v_{K,j}^{(s)}/\sum_k v_{K,k}^{(s)}$,
$v_{\mu,j}^{(s)}/\sum_k v_{\mu,k}^{(s)}$ and
$\{v_{K,j}^{(s)}+v_{\mu,j}^{(s)}\}/\sum_k\{v_{K,k}^{(s)}+v_{\mu,k}^{(s)}\}$,
and the table reports their posterior means. The combined share gives
the overall effective contribution of an expert. It is a derived
summary rather than a model parameter and omits covariance terms
among the summed components, so it remains nonnegative but is not an
exact decomposition of the latent variance.

Relative to the prior allocation, the posterior combined share of the covariate expert falls from $0.500$ to $0.301$, while the three networks rise jointly from $0.500$ to $0.699$. Within the structured mean, the covariate expert has the largest share at $0.342$, whereas
the \textcolor{black}{Network~3} expert has the smallest at $0.155$. The \textcolor{black}{Network~3} expert is more prominent in the covariance component, where its share is $0.264$. Across all experts, the structured mean has a mean additive share of $0.203$, with a $5\%$ to $95\%$ interval of $[0.027,0.517]$. The corresponding exact decomposition gives mean fractions of $0.190$ for the structured mean, $0.757$ for the GP component and $0.052$ for their interaction. None of the $21$ structured mean coefficients has a $90\%$ posterior interval excluding zero, so the mean block is interpreted jointly rather than one coefficient at a time.

\begin{table}[tbp]
\centering
\caption{Prior allocation and posterior effective expert shares.}
\label{tab:expert-weights}
\small
\begin{tabular}{lrrrr}
\toprule
\multirow{2}{*}{Expert} &
Prior &
Posterior &
Posterior &
Posterior \\[-0.6ex]
&
Combined Share &
Covariance share &
Mean share &
Combined share \\
\midrule
Covariates & 0.500 & 0.306 & 0.342 & 0.301 \\
Network 1 & 0.167 & 0.230 & 0.249 & 0.228 \\
Network 2 & 0.167 & 0.201 & 0.254 & 0.224 \\
Network 3 & 0.167 & 0.264 & 0.155 & 0.247 \\
\bottomrule
\end{tabular}

\end{table}

\subsection{Comparison with the zero mean fit}
\label{supp:mu0}

Refitting with $\boldsymbol\mu_j=\mathbf 0$ isolates the contribution
of the structured mean. Figure~\ref{fig:supp-traces-mu0} shows the
parameter traces for this fit, Table~\ref{tab:supp-parameters-mu0}
the posterior summaries and effective sample sizes of the sampled
parameters, and Table~\ref{tab:supp-ess-mu0} the same for the latent
field and the induced probabilities.
Five of the six screening candidates are recovered and the ranking is largely unchanged, but the probabilities are not. On those five \textcolor{black}{firms} the mean posterior probability falls from $0.609$ to $0.487$, \textcolor{black}{while $\widehat H_\eta$, the estimate based on $\eta$ under SCAR, falls from $6.26$ to $4.59$}, and the descriptive scores on the observed positives deteriorate. \textcolor{black}{For the posterior count $H$, $\mathbb E(H\mid\mathbf y)$ falls from $5.25$ to $3.24$, with posterior medians $5$ and $3$ and $80\%$ HPD intervals $[0,8]$ and $[0,5]$, respectively. The corresponding $80\%$ HPD intervals for the count based on $\eta$ under SCAR are $[0.66,9.64]$ and $[0.44,7.10]$.} Table~\ref{tab:full-mu0} summarizes the two fits, while Table~\ref{tab:supp-candidate-comparison} compares the screening candidates individually. The sixth candidate is not selected by the $\mu=0$ fit. Thus, the structured mean appears to increase confidence in cases already identified by the model rather than substantially changing their ranking, and Figure~\ref{fig:prob-difference} shows that this increase is selective.

\begin{table}[tbp]
\centering
\caption{Posterior summaries and effective sample sizes for the
PU-PoE model with $\mu=0$.}
\label{tab:supp-parameters-mu0}
\small
\setlength{\tabcolsep}{4pt}
\begin{tabular}{lrrrrrr}
\toprule
Parameter & Mean & SD & 5\% & Median & 95\% & ESS \\
\midrule
$l_{\mathrm{0}}$ & 2.895 & 0.472 & 2.206 & 2.844 & 3.735 & 1268.8 \\
$\sigma_{\mathrm{0}}$ & 1.673 & 0.526 & 0.924 & 1.609 & 2.635 & 638.7 \\
$l_{\mathrm{1}}$ & 0.312 & 0.018 & 0.283 & 0.312 & 0.341 & 1530.8 \\
$\sigma_{\mathrm{1}}$ & 8.519 & 2.710 & 4.758 & 8.146 & 13.447 & 1170.6 \\
$l_{\mathrm{2}}$ & 0.340 & 0.023 & 0.304 & 0.340 & 0.378 & 1389.9 \\
$\sigma_{\mathrm{2}}$ & 6.636 & 2.093 & 3.623 & 6.423 & 10.503 & 1009.3 \\
$l_{\mathrm{3}}$ & 0.335 & 0.020 & 0.304 & 0.334 & 0.369 & 1233.5 \\
$\sigma_{\mathrm{3}}$ & 8.460 & 3.929 & 3.798 & 7.544 & 16.199 & 461.2 \\
$\eta$ & 0.105 & 0.063 & 0.022 & 0.095 & 0.225 & 2773.1 \\
$\beta_0$ & $-3.390$ & 0.536 & $-4.280$ & $-3.390$ & $-2.515$ & 4124.3 \\
\bottomrule
\end{tabular}

\end{table}

\begin{table}[tbp]
\centering
\caption{Effective sample size summaries for the latent field and the
induced probabilities under the $\mu=0$ fit.}
\label{tab:supp-ess-mu0}
\small
\begin{tabular}{lrrrrrr}
\toprule
Quantity & Mean ESS & Median ESS & Min ESS & 5\% ESS & 95\% ESS & Max ESS \\
\midrule
$f$ & 1675.7 & 1527.4 & 710.6 & 967.0 & 2839.1 & 4301.5 \\
$\sigma(f)$ & 4314.4 & 4467.2 & 1969.7 & 3265.3 & 5000.0 & 5000.0 \\
\bottomrule
\end{tabular}

\end{table}

\begin{figure}[tbp]
\centering
\includegraphics[width=0.90\linewidth]
  {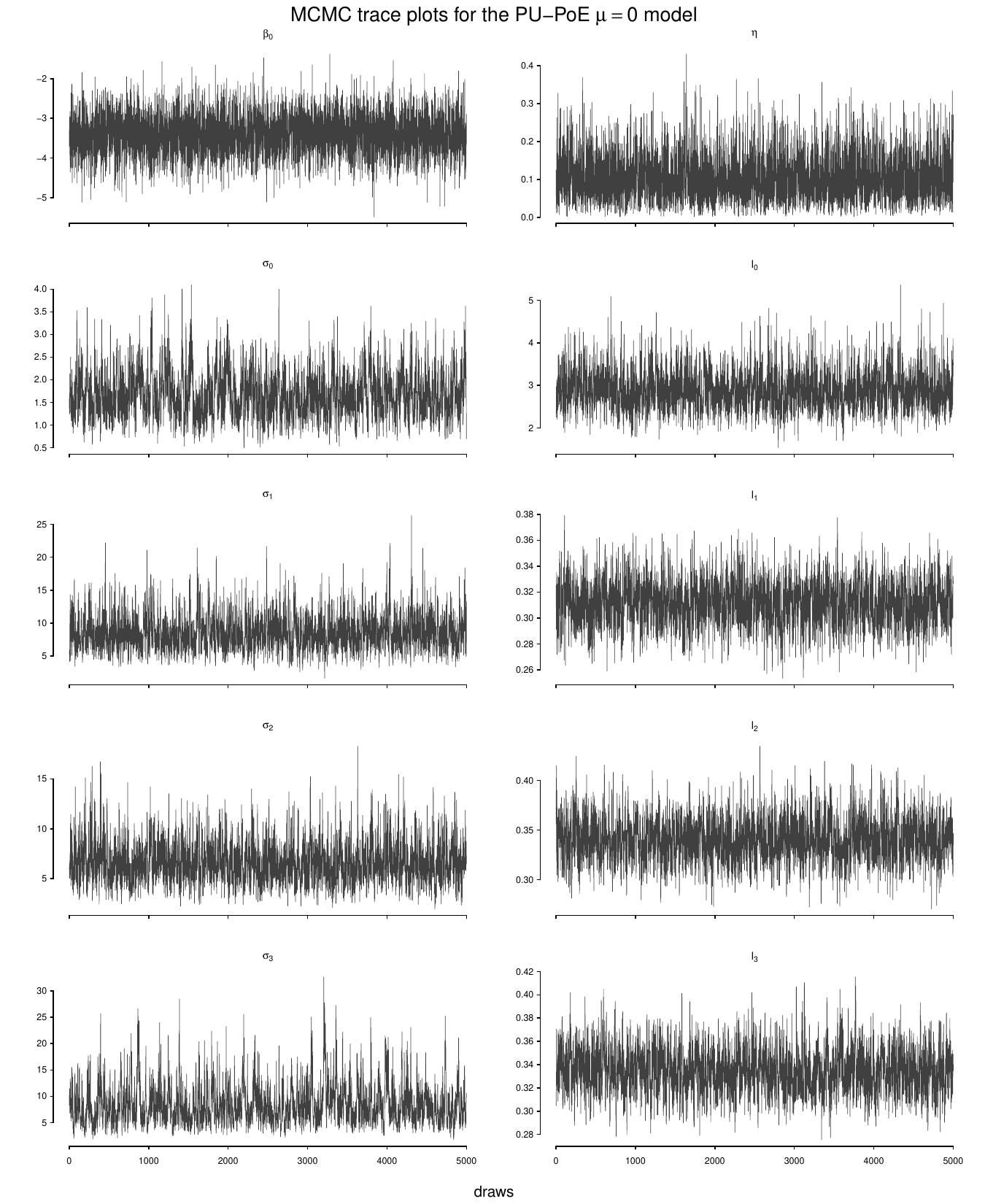}

\caption{Parameter trace plots for the PU-PoE model with $\mu=0$.}
\label{fig:supp-traces-mu0}
\end{figure}

\begin{table}[tbp]
\centering
\caption{\textcolor{black}{Full versus $\mu=0$ PU-PoE summaries. The rows for $H$ summarize the posterior count in \eqref{eq:H-posterior}, while $H_\eta$ is the count based on $\eta$ under SCAR defined in \eqref{eq:H-estimator}.} The log loss
and Brier entries are descriptive in sample scores on the $37$
observed positives.}
\label{tab:full-mu0}
\small
\begin{tabular}{lrr}
\toprule
Quantity & Full PU-PoE & PU-PoE $(\mu=0)$ \\
\midrule
$\mathbb E(\eta\mid\mathbf y)$ & 0.137 & 0.105 \\
$\mathbb E(H\mid\mathbf y)$ & 5.25 & 3.24 \\
Posterior median of $H$ & 5 & 3 \\
$80\%$ HPD for $H$ & $[0,8]$ & $[0,5]$ \\
$\widehat H_\eta$ & 6.26 & 4.59 \\
$80\%$ HPD for $H_\eta$ & $[0.66,9.64]$ & $[0.44,7.10]$ \\
Mean observed positive probability & 0.801 & 0.772 \\
Median observed positive probability & 0.900 & 0.849 \\
Minimum observed positive probability & 0.205 & 0.153 \\
Observed positive log loss & 0.259 & 0.310 \\
Observed positive Brier & 0.072 & 0.091 \\
\bottomrule
\end{tabular}

\end{table}

\begin{table}[tbp]
\centering
\caption{Candidate level comparison between the full and $\mu=0$
fits.}
\label{tab:supp-candidate-comparison}
\small
\begin{tabular}{lrrrrr}
\toprule
Candidate & Rank (full) & Rank ($\mu=0$) &
$\hat p_{\mathrm{full}}$ & $\hat p_{\mu=0}$ & Selected by both \\
\midrule
C1 & 1 & 1 & 0.720 & 0.676 & Yes \\
C2 & 2 & 3 & 0.637 & 0.477 & Yes \\
C3 & 3 & 4 & 0.605 & 0.414 & Yes \\
C4 & 4 & 2 & 0.557 & 0.501 & Yes \\
C5 & 5 & 5 & 0.524 & 0.368 & Yes \\
C6 & 6 & -- & 0.459 & -- & No \\
\bottomrule
\end{tabular}

\end{table}

\begin{figure}[tbp]
\centering
\includegraphics[width=0.9\linewidth]
  {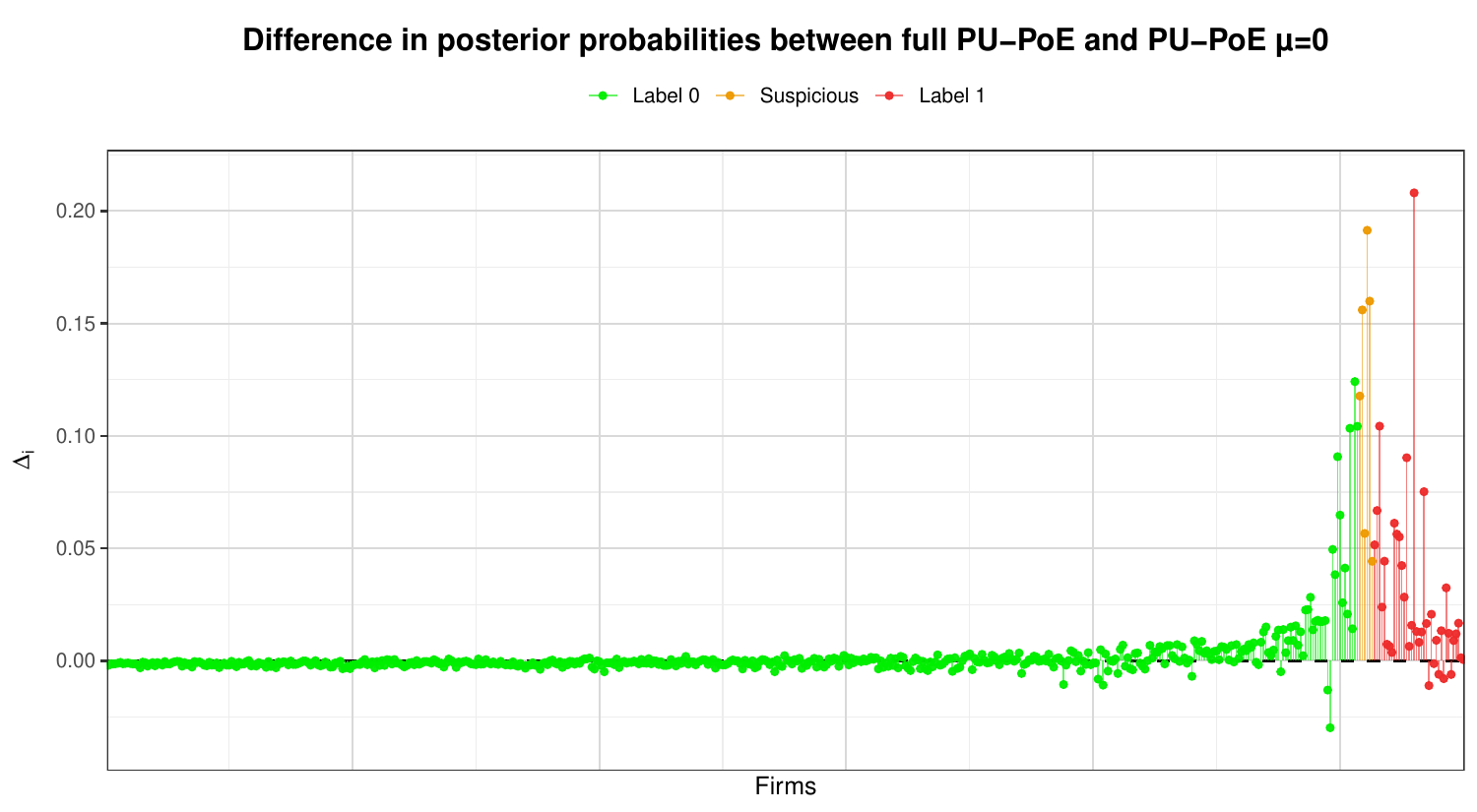}
  
\caption{Differences
$\hat p_{i,\mathrm{full}}-\hat p_{i,\mu=0}$ for all $550$ \textcolor{black}{firms},
ordered as ordinary unlabeled \textcolor{black}{firms}, the six candidates and the
$37$ observed positives.}
\label{fig:prob-difference}
\end{figure}

\FloatBarrier

\end{document}